\documentclass[11pt]{article}

\def\TALG{1}

\usepackage{bbm}
\usepackage{paralist}
\usepackage{array}
\usepackage[skip=0.5\baselineskip]{caption}
\usepackage{booktabs}
\def\conf{0}
\def\soda{0}

\usepackage{etex}
\usepackage[margin=2.54cm]{geometry}
\usepackage[dvipsnames,usenames]{color}
\usepackage{amsfonts,amsmath,amssymb,amsthm}

\usepackage{amsfonts,amsmath}
\usepackage{mathtools}

\usepackage[margin=10pt]{subfig}

\usepackage{paralist}
\usepackage{algorithmic,algorithm}
\usepackage{bm}
\usepackage{xspace}
\usepackage{xspace,prettyref}
\usepackage{color}
\usepackage{dsfont}
\usepackage{algorithmic}

\usepackage[colorlinks,urlcolor=blue,citecolor=blue,linkcolor=blue]{hyperref}

\newcommand{\confFigure}[1]
{
\ifnum\conf=1
	\begin{figure}[htb!]
		\fbox{
			\begin{minipage}{0.465\textwidth}
				\begin{raggedright}
					#1					
				\end{raggedright}
			\end{minipage}
		}
	\end{figure}
\else
	\begin{figure}[htb!]
	\fbox{
		\begin{minipage}{0.9\textwidth}
			\begin{raggedright}
				#1					
			\end{raggedright}
		\end{minipage}
	}
\end{figure}

\fi
}

\newcommand{\ignore}[1]{}

\newtheorem{theorem}{Theorem}

\newtheorem{lemma}[theorem]{Lemma}
\newtheorem{claim}[theorem]{Claim}

\newtheorem{definition}[theorem]{Definition}

\newcommand{\EX}{\mathbb{E}}

\renewcommand{\th}{^\textrm{th}}

\newcommand{\eqdef}{\stackrel{\rm def}{=}}
\def\eps{\epsilon}
\def\bar{\overline}

\newenvironment{proofof}[1]{\smallskip\noindent{\bf Proof of #1:}}%
        {\hspace*{\fill}$\Box$\par}

\newcommand{\om}{\overline{m}}

\newcommand{\Sec}[1]{\hyperref[sec:#1]{\S\ref*{sec:#1}}} 
\newcommand{\Eqn}[1]{\hyperref[eq:#1]{(\ref*{eq:#1})}} 
\newcommand{\Fig}[1]{\hyperref[fig:#1]{Fig.\,\ref*{fig:#1}}} 
\newcommand{\Tab}[1]{\hyperref[tab:#1]{Tab.\,\ref*{tab:#1}}} 
\newcommand{\Thm}[1]{\hyperref[thm:#1]{Theorem\,\ref*{thm:#1}}} 
\newcommand{\Lem}[1]{\hyperref[lem:#1]{Lemma\,\ref*{lem:#1}}} 
\newcommand{\Prop}[1]{\hyperref[prop:#1]{Prop.~\ref*{prop:#1}}} 
\newcommand{\Cor}[1]{\hyperref[cor:#1]{Corollary~\ref*{cor:#1}}} 
\newcommand{\Def}[1]{\hyperref[def:#1]{Definition~\ref*{def:#1}}} 
\newcommand{\Alg}[1]{\hyperref[alg:#1]{Alg.~\ref*{alg:#1}}} 
\newcommand{\Ex}[1]{\hyperref[ex:#1]{Ex.~\ref*{ex:#1}}} 
\newcommand{\Clm}[1]{\hyperref[clm:#1]{Claim~\ref*{clm:#1}}} 
\newcommand{\Assum}[1]{\hyperref[assump:#1]{Assumption~\ref*{assump:#1}}} 

\newcommand{\comp}{{\tilde{O}}\left(\frac{n}{\eps\sqrt{m}}\right) +
	\left(\frac{1}{\eps}\right)^{O(\log(1/\eps))}}

\def\poly{{\rm poly}}

\newcommand{\dist}[2]{\eps_{#1,#2}}

\newcommand{\davg}{\overline{d}}

\newcommand{\mnote}[1]{\marginpar{\tiny\bf
		\begin{minipage}[t]{0.5in}
			\raggedright#1
		\end{minipage}}}

\newcommand\numberthis{\addtocounter{equation}{1}\tag{\theequation}}

\newcommand{\mG}{\mathcal{G}}

\renewcommand{\Pr}{\mathrm{Pr}}

\def\nofigures{1}

\ifnum\nofigures=0
\usepackage{pstricks,tikz}
\usepackage{subcaption}
\usetikzlibrary{decorations.markings}
\usetikzlibrary{patterns}
\pgfdeclarelayer{background}
\pgfsetlayers{background,main}
\fi

\def\withcolors{0}
\ifnum\withcolors=1
\newcommand{\old}[1]{{\textcolor{gray}{ #1}} }
\newcommand{\changed}[1]{{\color{green}{#1}}}

\newcommand{\talya}[1]{{\color{blue}  #1}}
\newcommand{\strikeout}[1]{{\color{gray}  #1}}

\newcommand{\danachange}[1]{{\color{blue}#1}}
\newcommand{\tnew}[1]{{\color{red} #1}}

\else
\newcommand{\changed}[1]{{{#1}}}
\newcommand{\old}[1]{}

\newcommand{\talya}[1]{{\color{black}  #1}}
\newcommand{\strikeout}[1]{}

\newcommand{\danachange}[1]{#1}
\newcommand{\tnew}[1]{{  #1}}
\fi

\newcommand{\rG}[1][\ell]{G_{#1}(0)}
\newcommand{\Gmax}[1][\ell]{G_{#1}(\eps)}
\newcommand{\Umax}[1][\ell]{A_{#1}(\eps)}
\newcommand{\Ualpha}[1][\ell]{A_{#1}(0)}
\newcommand{\td}{\widetilde{d}}

\newcommand{\Erd}{\hyperref[Erd]{\color{black} \bf Estimate-Remaining-Degree}}
\newcommand{\Ehd}{\hyperref[Ehe]{\color{black} \bf Estimate-High-Edges}}
\newcommand{\Erle}{\hyperref[Erle]{\color{black} \bf Estimate-Remaining-Edges}}

\newcommand{\Enoe}{\hyperref[Enoe]{\color{black} \bf Estimate-Number-of-Edges}}
\newcommand{\AssignEdges}{\hyperref[Ae]{\color{black} \bf Assign-Edges}}
\newcommand{\IsActive}{\hyperref[IsAct]{\color{black} \bf Is-Active}}
\newcommand{\RecursiveIsActive}{\hyperref[RecIsAct]{\color{black} \bf Recursive-Is-Active}}
\newcommand{\Seau}{\hyperref[Seau]{\color{black} \bf Sample-Edge-Almost-Uniformly}}
\newcommand{\Iba}{\hyperref[Iba]{\color{black} \bf Is-Bounded-Arboricity}}
\newcommand{\Sle}{\hyperref[Sle]{\color{black} \bf Sample-Light-Edge}}
\newcommand{\Ibgm}{\hyperref[Ibgm]{\color{black} \bf Is-Bounded-Arboricity-Given-Edges-Estimate}}

\newcommand{\act}{active}

\newcommand{\sets}{\frac{300	}{\eps}}
\newcommand{\setq}{\frac{\alpha n}{\eps ^2 \cdot \bar{m}}\cdot (300\ln(4/\delta)) }		
\newcommand{\setr}{\frac{200\ln(1/\delta)}{\eps}}
\renewcommand{\dist}{\eps'}
\newcommand{\settt}{\frac{600}{\eps}}
		
\renewcommand{\emph}[1]{\textsf{#1}}

\newcommand{\YES}{{\sf Yes}}
\newcommand{\NO}{{\sf No}}
\newcommand{\MANY}{{\sf Many}}
\newcommand{\FEW}{{\sf Few}}

\newcommand{\tm}{\widetilde{m}_{\ell}}

\begin{document}

\title{Testing bounded arboricity}
	
	\author{
	Talya Eden\thanks{CSAIL at MIT. Email: {\tt talyaa01@gmail.com}.}
		\and
		Reut Levi\thanks{Efi Arazi School of Computer Science,  The Interdisciplinary Center, Israel. Email: {\tt reut.levi1@idc.ac.il}.}
			\and
		Dana Ron\thanks{School of Electrical Engineering, Tel Aviv University. Email: {\tt danaron@tau.ac.il}.}
}

\maketitle
	
\begin{abstract}
	In this paper we consider the problem of testing whether a graph has bounded arboricity.
The class of graphs with bounded arboricity includes many important graph families (e.g.,  planar graphs and randomly generated preferential attachment graphs). Graphs with bounded arboricity have been studied extensively in the past, in particular, since for many problems they allow for much more efficient algorithms and/or better approximation ratios.

We present a tolerant tester in the general-graphs model. The general-graphs model allows access to degree and neighbor queries, and the distance is defined with respect to the actual number of edges. Namely, we say that a graph $G$ is $\epsilon$-close to having arboricity $\alpha$, if by removing at most an $\epsilon$-fraction of its edges we can obtain a graph $G'$ that has arboricity $\alpha$, and otherwise we say that $G$ is $\epsilon$-far.
Our algorithm distinguishes between graphs that are $\epsilon$-close to having arboricity $\alpha$ and graphs that are $c \cdot \epsilon$-far from having arboricity $3\alpha$, where $c$ is an absolute small constant.
The query complexity and running time  of the algorithm are $\tilde{O}\left(\frac{n}{\epsilon\sqrt{m}}\right) + \left(\frac{1}{\epsilon}\right)^{O(\log(1/\epsilon))}$ where $n$ denotes the number of vertices and $m$ denotes the number of edges  (we use the notation $\tilde{O}$ to hide poly-logarithmic factors in $n$).  In terms of the dependence on $n$ and $m$ this bound is optimal up to poly-logarithmic factors since $\Omega\left(\frac{n}{\sqrt{m}}\right)$ queries are necessary.

\end{abstract}

\section{Introduction}

The  arboricity of a graph is defined as the minimum number of forests into which its edges can be partitioned.
This measure is equivalent (up to a factor of 2) to the maximum average degree in any subgraph~\cite{NW1,Tutte,NW2} and to the degeneracy of the graph.\footnote{The degeneracy of a graph $G$ is the smallest integer $k$ such that in every subgraph of $G$ there is vertex of degree at most $k$.
The arboricity of $G$ is upper bounded by its degeneracy and  the degeneracy is less than twice the arboricity.}
Hence, the arboricity of a graph can be viewed as a measure of its density ``everywhere''.
The 
class of graphs with bounded arboricity includes many important families of graphs, e.g., all minor-closed graph classes such as planar graphs, graphs of bounded treewidth and graphs of bounded genus.
Furthermore, graphs in this class are not restricted to being minor-free (for some fixed minor).
In fact, graphs over $n$ vertices with arboricity $2$ may have a $K_{\sqrt{n}}$-minor.
In the context of social networks, 
graphs that are generated according to evolving graph models such as the Barab{\'{a}}si-Albert Preferential
  Attachment model~\cite{BA99} have bounded arboricity.
  For various graph optimization problems, it is known that better approximation ratios and faster algorithms exist for graphs with bounded arboricity
 (e.g.,~\cite{CN85,GV08, ELS13, BU17} 
 and~\cite{BE,LPW13} in the distributed setting), and several NP-hard problems such as {\sc Clique}, {\sc Independent-Set} and {\sc Dominating-Set} become fixed-parameter tractable~\cite{AG09, ELS13}.

In this work we address the problem of testing whether a graph has bounded arboricity.
That is, we are interested in an algorithm that with high constant probability accepts graphs that have  arboricity bounded by a
given $\alpha$, and rejects graphs that are relatively far from having slightly larger arboricity
(in the sense that relatively many edges should be removed so that the graph will have such arboricity).
In fact, as explained precisely next, we solve a {\em tolerant\/}~\cite{PRR} version of this problem in which we  accept graphs that are only close to having arboricity $\alpha$. Furthermore, our result is in what is known as the
{\em general-graphs model\/}~\cite{PR}, where there is no upper bound on the maximum degree in the graph, and
distance to having a property is measured with respect to the number of edges in the graph.
As we discuss in more detail in
Section~\ref{subsec:related}, almost  all previous results on testing related
bounded graph measures assumed the graph has bounded degree.



\subsection{Our result}\label{subsec:result}
Let $G = (V,E)$ be a graph with $n$ vertices and $m$ edges.
We assume that for any given vertex $v \in V$, it is possible to query for its degree, $d(v)$, as well as query
for its $i^{\rm th}$ neighbor for any $1 \leq i \leq d(v)$.
~\footnote{We note that the ordering of the neighbors of vertices is arbitrary and that a neighbor query to vertex $v$ with $i > d(v)$ is answered by a special symbol. Observe that a degree query to $v$ can be replaced by $O(\log d(v))$ neighbor queries.}\textsuperscript{,}\footnote{Usually the general-graphs model also allows for pair queries, however our algorithm does not require them, and the lower bounds holds also when allowing them.}
We say that $G$ is {\em $\eps$-close\/} to having arboricity $\alpha$ if at most $\eps\cdot m$ edges should be removed from $G$ so that the resulting graph will have arboricity at most $\alpha$. Otherwise, $G$ is {\em $\eps$-far\/} from having arboricity $\alpha$.

We present an algorithm that, given query access to $G$ together with parameters $n$, $\alpha$ and $\epsilon$
 distinguishes with high constant probability 
 between the case that $G$ is $\eps$-close to having arboricity at most $\alpha$, and the case 
that $G$ is
 $c\cdot \eps$-far from having arboricity $3\alpha$ for an absolute constant\footnote{The constant we achieve is 20.
 For the sake of simplicity and clarity of the algorithm and its analysis, we did not make an effort to minimize this constant.}
  $c$.
The query complexity and running time of the algorithm are
$$\comp$$
 in expectation.
\subsection{Discussion of the result}\label{subsec:discuss}

\changed{In this subsection we discuss several aspects of our result, as well as some variants.
The variants are summarized in Table~\ref{tab:variants}.
}

\paragraph{The tightness of the complexity bound.}
If we consider the complexity of the algorithm as a function of $n$ and $m$ (ignoring the dependence on $\eps$), we get that it is 
\changed{$\tilde{O}\left(\frac{n}{\sqrt{m}}\right)$.}
We observe that this complexity is essentially tight, even for non-tolerant algorithms
\changed{(i.e., that distinguish between the case that $G$ \emph{has} arboricity at most $\alpha$, and the case that $G$ is far from having arboricity $3\alpha$)}.
To be precise, for constant $\eps$, $\Omega\left(\frac{n}{\sqrt{m}}\right)$ queries are necessary for any  algorithm that is not provided with any information regarding $m$, or even when it is provided with a constant factor estimate of $m$ (e.g., a factor-2 estimate).
\changed{The lower-bound construction is based on two families of graphs, where graphs in one family have arboricity at most $\alpha$ and graphs in the other family are far from having arboricity at most $3\alpha$. The graphs in the second family have a slightly larger number of edges, where, roughly speaking, these edges are the source of the distance to \tnew{bounded} arboricity, and they belong to a relatively small
``hidden'' subgraph (over $O(\sqrt{m})$ vertices). Other than this small subgraph, graphs in the two families have identical structure.}

If the algorithm is provided with $m$ (or a very precise estimate, i.e., within $(1\pm \eps/c)$ for $c>1$), then \changed{we cannot use this lower-bound construction, as graphs in both families must have the same (or almost the same) number of edges.}
However, in such a case \changed{we can modify the construction (so that graphs in both families have exactly the same number of edges) and obtain a lower bound of
$\Omega\left(\frac{n\alpha}{m}\right)$ (graphs in the two families now differ on subgraphs of size $O(m/\alpha)$).
Furthermore, we show that this number of queries is sufficient (when the algorithm is provided with a precise estimate of $m$).}
\changed{Note that by~\cite{NW1, NW2}, $\alpha \leq \sqrt{m}$, so that
$\frac{n\alpha}{m} \leq \frac{n}{\sqrt{m}}$.}

\paragraph{Bounded-degree graphs.}
Suppose first that we are given an upper bound $d$ on the maximum degree in $G$, and let $\davg = 2m/n$ denote the average degree. Then we can slightly modify the algorithm so that the term $\tilde{O}(n/\sqrt{m})$ in the complexity of the algorithm is replaced by $d/\davg$.

The above statement is for the case that  distance to having the property is measured (as defined in the general-graphs model), with respect to $m$ (and we only assume that the algorithm is provided with additional information regarding the maximum degree in the graph). If we consider the bounded-degree model~\cite{GR02}, in which not only do we get $d$ as input, but in addition distance is measured with respect to $d\cdot n$ (which is an upper bound on $m$), then our algorithm can be slightly modified so that its complexity depends {\em only on $1/\eps$\/} (and the dependence is quasi-polynomial).



\paragraph{Access to random edges.}
	In the case that the algorithm is given access to uniformly (or almost uniformly) distributed random edges, it can again be slightly modified to  run in time quasi-polynomial in $1/\eps$.
	\changed{This is true since the term $\tilde{O}\left(\frac{n}{\eps \sqrt{m}}\right)$
	in the running time arises from sampling edges almost uniformly at random when the number of edges is unknown.}

\paragraph{Expected complexity.}
The reason that the query complexity and running time are in expectation is due to the need of the edge sampling procedure to obtain an estimate for the number of edges. If such a (constant factor) estimate is provided to the algorithm, then the \changed{upper} bound on the complexity of the algorithm
\changed{always holds}.

\paragraph{$\alpha$ vs. $3\alpha$.}
Our algorithm distinguishes between the case that the graph is close to having arboricity at most $\alpha$ and the case that it is far from having arboricity at most $3\alpha$. The constant $3$ can be reduced to $2+\eta$
    at a cost that depends (exponentially) on $1/\eta$, but we do not know how to avoid this cost and possibly go below a factor of $2$. However, in some cases this constant may  not be significant.
    For example, suppose we want to know whether, after removing a small fraction of the edges, we can obtain a graph $G'$    with bounded arboricity so that we can run an optimization algorithm on $G'$
    (or possibly on $G$ itself),
    whose complexity depends polynomially on the arboricity of $G'$.
    In such a case, the difference between $\alpha$ and $3\alpha$ is inconsequential.

\paragraph{Two-sided error vs. one-sided error.}
Our algorithm has two-sided error, and we observe that \changed{for $\alpha\geq 2$} every one-sided error algorithm must perform $\Omega(n)$ queries. We note that for $\alpha=1$, there exists a one-sided error  algorithm for testing cycle-freeness~\cite{CGRSSS} that performs  $\tilde{O}(\sqrt n)$ queries (and these many queries are necessary~\cite{GR02}).

\paragraph{Dependence on $\eps$.}
In the second term of the complexity of our algorithm there is a quasi-polynomial dependence on $1/\eps$. It is
an open problem whether this dependence can be reduced to polynomial. 

\tnew{
	\renewcommand{\arraystretch}{1.8}

	\begin{table}
		\begin{tabular}{ | m{5.5cm} | m{5.5	cm} | m {2 cm} |}
			\hline
			\multicolumn{3}{|c|}{ \vspace{.05ex} \large\bf Summary of our result and variants }  \\
			\hline\hline\vspace{.000ex}
			\tnew{General-graphs model, two-sided error, unknown $m$, distance w.r.t. $m$} & $\comp$  & Theorem~\ref{thm:test-arb} \\ \cline{2-3}
 					   & $\Omega\left(\frac{n}{\sqrt{\eps m}}\right)$ \old{(const. $\eps$)} & Claim~\ref{clm:lb1}
 			\\ \hline
			Known $m$ (or $(1\pm \eps/c)$-estimate)   &    ${O}\left(\frac{n\alpha}{\eps^3 m}\right) +
			\left(\frac{1}{\eps}\right)^{O(\log(1/\eps))} $ & Theorem~\ref{thm:test-given-m} \\
			 & $\Omega\left(\frac{n\alpha}{\eps m}\right)$ & Claim~\ref{clm:lb2} \\
			 \hline
			Upper bound $d$ on max-degree &   $O\left(\frac{d}{\eps \bar{d}}\right)
    \danachange{+ \left(\frac{1}{\eps}\right)^{O(\log(1/\eps))}}$ & Section~\ref{var:bounded-deg} \\
			\hline
			Upper bound $d$ on max-degree and distance w.r.t. $n\cdot d$ & $(\frac{1}{\eps})^{O(\log(1/\eps))}$ & Section~\ref{var:bounded-dist} \\ \hline
			Access to random edges &  $ (\frac{1}{\eps})^{O(\log(1/\eps))}$  & Section~\ref{var:edge-queries} \\ \hline
				 One-sided error, \danachange{$\alpha\geq 2$} &  $\Omega(n)$ & Claim~\ref{clm:one-sided} \\ \hline
		\end{tabular}
		\caption{
			\tnew{Unless explicitly specified otherwise, all results refer to two-sided error algorithms in the general-graphs model (so that the distance is with respect to the number of edges $m$, and $m$ is unknown).
For each variant we specify the difference(s) in terms of the model (or error type).}}
		\label{tab:variants}
	\end{table}
}

\subsection{The algorithm}\label{subsec:techniques}

\paragraph{Challenges of designing algorithms in the general-graphs models \danachange{and our approach}}
	The general-graphs model presents 
\danachange{several} challenges. Indeed, 
\danachange{relatively} few 
\danachange{properties}
were studied in this model.  
\danachange{These include} 
bounded diameter and connectivity~\cite{PR},  bipartiteness~\cite{KKR04} (\cite{CMOS11} for planar graphs), triangle-freeness~\cite{AKKR-tri,Gug06,Ras06}, cycle-freeness~\cite{CGRSSS} (with one-sided error)
and $k$-path-freeness~\cite{IY14} 
(when random edge queries are also allowed).
	Since this model allows for unbounded degrees, it could be the case
that the large distance to having the property is due to a relatively dense subgraph that resides on a small set of vertices. As a consequence, observing  this subgraph might be costly.
\danachange{Furthermore,} since the degrees are unbounded, it is not clear how to efficiently explore the graph.
	Previous techniques in this model include bounded-size BFS~\cite{PR} (that is, performing a BFS until a predetermined fixed number of vertices is discovered), random walks~\cite{CMOS11,IY14} and analysis based on color coding~\cite{IY14}.
	
\danachange{We apply a different} technique to meet the above challenges.
In our analysis, we characterize a ``special'' set of edges, \danachange{$S$}, such that  if the graph is $\eps$-close to having bounded arboricity, then 
\danachange{$|S|$} is small, and if the graph is far from having bounded arboricity, then 
\danachange{$|S|$}  is large.
\danachange{The question now is how to decide whether an edge belongs to $S$ or not.}
\danachange{We show} that given an edge  in the graph, we can
\danachange{perform a certain ``approximate decision'' regarding membership in $S$, which suffices for our purposes}.
In order to do so, we use a
	 procedure that recursively samples a subset of the neighbors of 
\danachange{a given} vertex, until it reaches the maximal depth of the recursion. At this point it returns a deterministic answer based on the degree of the vertex, which propagates up the recursion tree. The resulting queries of this process can also be viewed as a randomized BFS of bounded degree and depth.
\danachange{We next provide some more details on the algorithm.}

\paragraph{The algorithm.}
Our starting point is a simple (non-sublinear and deterministic) algorithm that is similar to the distributed forest decomposition algorithm of Barenboim and Elkin~\cite{BE}. This algorithm works in $\ell = O(\log(1/\eps))$ iterations, where in each iteration it assigns edges to a subset of the vertices, and the vertices that are assigned edges become ``inactive''.
We show that if the graph has arboricity at most $\alpha$, then, when the algorithm terminates,  the number of edges between remaining active vertices, whose set we denote here by $A_\ell$, is relatively small. On the other hand, if the graph is sufficiently far from having arboricity at most $3\alpha$, then the number of  edges between vertices in $A_\ell$ is relatively large.

Given this statement regarding the number of remaining edges between vertices in $A_\ell$,
our algorithm estimates the number of such remaining edges.
To this end we devise a procedure for deciding whether a given vertex $v$ belongs to $A_\ell$.
This can  be done by emulating the deterministic algorithm on the distance-$\ell$ neighborhood of $v$. However, such an emulation may require a very large number of queries (as the maximum degree in the graph is not necessarily bounded). Instead, we perform a certain approximate randomized emulation of the deterministic algorithm, which is much more query efficient. While this emulation does not exactly answer whether or not $v\in A_\ell$, it gives an approximate answer
that suffices for our purposes (see Lemma~\ref{lem:is-active} for the precise statement).

The high-level idea is that given a vertex $v$, we select a random subset of its neighbors. We then recursively run the  procedure to decide for each of these neighbors whether it belongs to $A_{\ell-1}$ (the vertices that remain active after $\ell-1$ iterations). This recursive process defines a (random) tree with $\ell$ levels. For each vertex in the tree we decide if it is active or not according to the fraction of active children it has in the tree, and for the leaf vertices we simply decide according to their degree.

When analyzing the correctness of this procedure for a vertex $v$ we need to take into account two sources of error.
 The first is due to a possible bias in the selection of its sample of neighbors. That is, even if we had an oracle that always answered correctly for a vertex $u$ whether it belongs to $A_{\ell-1}$, we might still err in our decision regarding whether or not $v$ belongs to $A_{\ell}$. The second source of error is due to incorrect answers on $v$'s neighbors. In other words, we need to analyze how errors propagate (and accumulate) up the recursion tree.

 In addition, we  upper bound the total size of the tree (which determines the query complexity and running time of the procedure).

\subsection{Related Work}\label{subsec:related}
In what follows, when we say ``testing'' we mean ``property testing'', as defined earlier. That is, distinguishing between objects that have a property and objects that are far from having the property.

Most of the 
property testing results related to this work
are in the bounded-degree model~\cite{GR02}. Recall that in this model
the algorithm has the same query access to the graph as we consider, but it is also given an upper bound, $d$,
on the maximum degree in the graph, and distance is measured with respect to $d\cdot n$ (rather than the actual number of edges, $m$), so that it is less stringent. As noted in
Section~\ref{subsec:discuss},
an adaptation of our algorithm to the (``easier'') bounded-degree model achieves complexity
that is quasi-polynomial in $1/\eps$ (and independent of $n$).
As we discuss next, in the bounded-degree model there are several results on testing
whether a graph excludes  specific fixed minors as well as results on testing minor-closed properties in general.
In what follows we assume that $d$ is a constant, since in some of these works this assumption is made
(so that no explicit dependence on $d$ is stated).

Goldreich and Ron~\cite{GR02} provide an algorithm for testing if a graph is cycle-free, namely, excludes  $C_3$-minors,
where the complexity of the algorithm is $O(1/\eps^{3})$.
Yoshida and Ito~\cite{YI15} test outerplanarity (excluding $K_4$-minors and $K_{2,3}$-minors) and if a graph is a cactus (excluding a diamond-minor) in time that is polynomial in $1/\eps$.
Benjamini, Schramm, and Shapira~\cite{BSS08} showed that any minor-closed property can be tested in time that
depends only on $1/\eps$ (where the dependence may be triply-exponential).
Hassidim et al.~\cite{HKNO09} introduced a general tool, a partition oracle, for locally partitioning  graphs
 that belong to certain families of graphs, into small parts with relatively few edges between the parts.
A partition oracle for a family of graphs implies a corresponding (two-sided error) tester for membership in this family.
Hassidim et al.~\cite{HKNO09} designed partition oracles for hyperfinite classes of graphs and minor-closed classes of graphs.
One of the implications of their work is improving the running time of testing minor-closed properties from triply-exponential in $\poly(1/\eps)$ to singly exponential in $\poly(1/\eps)$.
Levi and Ron later improved the running time of the partition oracle for minor-closed classes of graphs to quasi-polynomial in $1/\eps$~\cite{LR15}.
Edelman et al. designed a partition oracle for graphs with bounded treewidth~\cite{EHNO11} whose query and time complexity are polynomial in $1/\eps$.
Newman and Sohler~\cite{NS13} extended the result of~\cite{HKNO09} and showed that every hyperfinite property
(i.e., property of hyperfinite graphs) is testable in  time that is independent of the size of the graph.

All the aforementioned testing algorithms have two-sided error (and this is also true of our algorithm). Czumaj et al.~\cite{CGRSSS} study
the problem of one-sided error testing of $C_k$-minor freeness and tree-minor freeness. For cycle-freeness ($C_3$-minor freeness)
 they give a one-sided error testing algorithm whose complexity is $\tilde{O}(\sqrt{n}\cdot \poly(1/\eps))$ (for $k>3$ there is an exponential dependence on $k$). They show that the dependence on $\sqrt{n}$ is tight for any minor that contains a cycle.
 On the other hand, for tree-minors they give an algorithm whose complexity is $\exp((1/\eps)^{O(k)})$,
 where $k$ is the size of the tree (so that the complexity is independent of $n$).

 Finally we discuss 
 results in the 
 \danachange{general-graphs} model 
 \danachange{that are  related to our result}.
  Czumaj et al.~\cite{CGRSSS} show that their result for cycle-freeness extends to the 
  \danachange{general-graphs} model,
 where the complexity of the algorithm is  $\tilde{O}(\sqrt{n}\cdot \poly(1/\eps))$.
 Iwama and Yoshida~\cite{IY14} consider an augmented model that allows random edge sampling. In this augmented model they provide several testers for parameterized properties including $k$-path freeness whose complexity is independent of the size of the graph.

\subsection{Organization}
Following some basic preliminaries in Section~\ref{sec:prel}, we give the aforementioned ``Edge-assignment algorithm'' in Section~\ref{sec:det-assign}.
In Section~\ref{sec:test-arb} we present our testing algorithm. The lower bounds mentioned in 
Section~\ref{subsec:discuss}
are provided in Section~\ref{sec:lower-bound} and the variants of our algorithm (e.g., in the bounded-degree model)
appear in Section~\ref{sec:variations}. A procedure for estimating what we refer to as the \emph{$\eps$-corrected arboricity} of a given graph (see Definition~\ref{def:corrected}) appears in Section~\ref{sec:corrected}.
In the appendix we describe an improved variant of our algorithm for the case that a precise estimate of the number of
edges is given to the algorithm.

\section{Preliminaries}\label{sec:prel}

For an integer $k$, let $[k]\eqdef \{1,\dots,k\}$.
For an undirected simple graph $G = (V,E)$ let $n = |V|$ and $m=|E|$. For each vertex $v\in V$, let
$d(v)$ denote its degree.

We assume there is query access to the graph in the form of {\em degree\/} queries and {\em neighbor\/} queries.
That is, for any vertex $v\in V$, it is possible to perform a query to obtain $d(v)$, and for any $v$ and $i \in [d(v)]$, it is possible to perform a query to obtain the $i^{\rm th}$ neighbor of $v$ (where the order over neighbors is arbitrary).
If $i> d(v)$, then a special symbol is returned.

\begin{definition}[Distance]
For a property ${\mathcal P}$ of graphs, and a parameter $\eps \in [0,1]$, we say that a graph $G$ is {\em $\eps$-far\/}
from (having) the property ${\mathcal P}$, if more than $\eps\cdot m$ edge modifications on $G$ are required so as to obtain
a graph that has the property ${\mathcal P}$.
\end{definition}

\begin{definition}[Arboricity]
The {\sf arboricity} of a graph $G = (V,E)$ is the minimum number of forests into which its edges can be partitioned.
We denote the arboricity of $G$ by $\alpha(G)$.
\end{definition}
By the work of Nash-Williams~\cite{NW1, NW2}, for every graph $G = (V,E)$,
\ifnum\conf=0
\begin{equation}\label{eq:Nash-Williams}
\alpha(G) \;= \; \max_{S\subseteq V} \left\lceil \frac{|E(S)|}{|S|-1} \right\rceil\;,
\end{equation}
\else
$\alpha(G) \;\leq \; \max_{S\subseteq V} \left\lceil \frac{|E(S)|}{|S|-1} \right\rceil$,
\fi
where $E(S)$ denotes the set of edges in the subgraph induced by $S$.

Let $\exp(x) \eqdef e^{x}$, and for a random variable $\chi$, we use $\EX[\chi]$ to denote its expected value.

\ifnum\conf=0
We make use of Hoeffding's inequality \cite{hoeffding}, stated next.
For $i=1, \ldots, s$, let $\chi_i$ be a $0/1$-valued random variable, such that $\Pr[\chi_i=1] = \mu$.
Then for any $\gamma \in (0,1]$,
\[ \Pr\left[\frac{1}{s}\sum\limits_{i=1}^s \chi_i > \mu+\gamma \right] < \exp\left( -2\gamma^2 s\right) \]
and
\[
			\Pr\left[\frac{1}{s}\sum\limits_{i=1}^s \chi_i < \mu-\gamma \right] < \exp\left( -2\gamma^2 s\right)\;.
\]

We also apply the following version of  the multiplicative Chernoff bound \cite{chernoff}.
For $i=1, \ldots, s$, let $\chi_i$ be a random variables taking values in $[0,B]$, such that  $\EX[\chi_i]=\mu$. Then
for any $\gamma \in (0,1]$,
\[ \Pr\left[\frac{1}{s}\sum\limits_{i=1}^s \chi_i > (1+ \gamma)\mu \right] < \exp\left( -\frac{\gamma^2 \mu s}{3B}\right)
\]
and
\[\Pr\left[\frac{1}{s}\sum\limits_{i=1}^s \chi_i < (1- \gamma)\mu \right] < \exp\left( -\frac{\gamma^2 \mu s}{2B}\right)\;.\]
\fi

\section{A Deterministic Edge-Assignment Algorithm}	\label{sec:det-assign}

In this section we describe a deterministic  algorithm that, given as input a graph $G=(V,E)$,  assigns edges to vertices.
The algorithm works iteratively, where in each iteration it assigns edges to a new subset of vertices.
The algorithm is provided with parameters that determine an upper bound on the number of edges that are assigned to each vertex (where an edge may be assigned to both of its endpoints). \changed{The number of edges assigned to each vertex is at most $3\alpha$ plus a small fraction of its original degree. This fraction is determined by one of the parameters ($\gamma$).} When the algorithm terminates, some edges may remain unassigned (and some vertices may not have been assigned any edges).
This algorithm (when viewed as a distributed algorithm) is a variant of the algorithm by Barenboim and Elkin~\cite{BE} for finding a forest decomposition in graphs with bounded arboricity.

\ifnum\TALG=0
\confFigure{
	{\bf Assign-Edges$(G,\alpha, \eps, \gamma)$} \label{Ae}
	\smallskip
	\begin{compactenum}
		\item $G_0(\gamma) = G$, $A_0(\gamma) = V$.
		\item For $i=1$ to  $\ell\eqdef                           \lceil\log_{6/5}(1/\eps)\rceil$ do:
		\begin{compactenum}
			\item Let $B_i(\gamma)$ be the set of vertices $v\in V$ whose  degree in $G_{i-1}(\gamma)$ is
			at most $3\alpha + \gamma\cdot d(v)$.
			\item Assign
			each vertex $v \in B_i(\gamma)$ the edges incident to it in $G_{i-1}(\gamma)$.
			\item Let $A_i(\gamma) = A_{i-1}(\gamma)\setminus B_i(\gamma)$, and let $G_i(\gamma)$ be the graph induced by $A_i(\gamma)$.
		\end{compactenum}
	\end{compactenum}
}
\else
\begin{algorithm}
	\caption{Assign-Edges$(G,\alpha, \eps, \gamma)$}
	\label{Ae}
\begin{algorithmic}[1]
	
		\STATE  $G_0(\gamma) = G$, $A_0(\gamma) = V$.
		\FOR{$i=1$ to $\ell\eqdef \lceil\log_{6/5}(1/\eps)\rceil$}
			\STATE Let $B_i(\gamma)$ be the set of vertices $v\in V$ whose  degree in $G_{i-1}(\gamma)$ is
			at most $3\alpha + \gamma\cdot d(v)$.
			\STATE Assign
			each vertex $v \in B_i(\gamma)$ the edges incident to it in $G_{i-1}(\gamma)$.
			\STATE Let $A_i(\gamma) = A_{i-1}(\gamma)\setminus B_i(\gamma)$, and let $G_i(\gamma)$ be the graph induced by $A_i(\gamma)$.
		\ENDFOR
\end{algorithmic}
\end{algorithm}
\fi

The algorithm \AssignEdges\ is provided with
3 parameters: $\alpha$, $\eps$ and $\gamma$.
It might be useful to first consider its execution with  $\gamma=0$.
The role of
$\gamma$
will become clear subsequently (when we describe our testing algorithm
and its relation to \AssignEdges).
In the case of $\gamma=0$, in every iteration, each vertex with degree at most $3\alpha$ in the current graph is assigned all its incident edges in this graph. The initial graph is $G$, and at the end of  iteration $i$, the vertices that are assigned edges, denoted $B_i(\gamma)$ in the algorithm\footnote{While $B_i(\gamma)$ depends also on $\alpha$, we shall want to refer to these sets when the algorithm is invoked with
	the same value of $\alpha$ but with different values of $\gamma$. Hence, only $\gamma$ appears explicitly in the notation.}, together with the edges assigned to them, are removed from the graph.
Once vertices are assigned edges, we view them as becoming {\em inactive\/}.
We use the notation $A_i(\gamma)$ for the vertices that are still {\em \act} 
at the end of iteration $i$.
Observe that by the definition of the algorithm, for $\gamma_1 \leq \gamma_2$, we have that
$A_i(\gamma_2) \subseteq A_i(\gamma_1)$ for every iteration $i$, and hence $G_i(\gamma_2)$ is a subgraph of $G_i(\gamma_1)$.

In the next
\changed{two lemmas}
we upper bound the number of edges in $G_\ell(0)$ when $G$
is close to having arboricity $\alpha$, and we lower bound the number of edges in $G_\ell(\gamma)$
(which is a subgraph of $G_\ell(0)$), when $G$ is far from having arboricity $3\alpha$.\footnote{The term $3\alpha$ can be improved to $(2+\eta)\alpha$ for any $\eta > 0$ by increasing the number of iterations by a factor of $1/\eta$.}

\begin{lemma}\label{lem:det-assign-ub}
If $G$ is $\eps$-close to having arboricity at most $\alpha$, then
$|E(G_\ell(0))| \leq 5\eps m$.
\end{lemma}

\begin{proof}
By the premise of the lemma, $G$ is $\eps$-close to having arboricity at most $\alpha$.
This implies that for each of its subgraphs, $G'$, there is a subset of at most $\eps m$ edges whose removal makes $G'$ have arboricity at most $\alpha$. In particular this is true for the subgraphs
$G_i(0)$ defined by the algorithm, for $i= 1,\dots,\ell$. Denoting by $m_i$  the number of edges in $G_i(0)$, we have
that for every $i \in [\ell]$,
\begin{equation}\label{eq:mi-ub}
m_i \;\leq\; \alpha \cdot|A_i(0)|+\eps m\;.
\end{equation}
By the definition of $A_i(0)$, each vertex $v\in A_i(0)$ has degree greater than
$3\alpha$ in $G_{i-1}(0)$. It follows that
\begin{equation}\label{eq:mi-one-lb}
m_{i-1}\geq 3\alpha |A_i(0)|/2\;.
\end{equation}
Suppose that $|A_i(0)|>  4\eps m/\alpha$ (so that $\eps m < \alpha|A_i(0)|/4$).
The upper bound on $m_i$ in Equation~(\ref{eq:mi-ub}) implies that $m_i \leq 5\alpha|A_i(0)|/4$.
combining this  with the lower bound on $m_{i-1}$ in Equation~(\ref{eq:mi-one-lb})
we get that
$\frac{m_i}{m_{i-1}}\leq 5/6.$

\danachange{Therefore}, in every iteration of \AssignEdges\ in which $|A_i(0)| > 4\eps m /\alpha$, the number of edges in the graph decreases by a multiplicative factor of $5/6$. On the other hand, if
$|A_i(0)|\leq 4\eps m /\alpha$, then,
\danachange{applying Equation~(\ref{eq:mi-ub}) with this upper bound on $|A_i(0)|$, we get that $m_i  \leq 5\eps m$.}
Hence, after at most
$\lceil\log_{6/5}(1/\eps)\rceil$  iterations, there are at most $5\eps m$ edges between active vertices.
\end{proof}

\begin{lemma}\label{lem:det-assign-lb}
If $G$ is $\dist$-far 
 from having arboricity $3\alpha$, then
 $|E(G_\ell(\gamma))| > (\dist\ -2\gamma) m$.
\end{lemma}

\begin{proof}
Assume, contrary to the claim, that
$|E(G_\ell(\gamma))| \leq (\dist -2\gamma)m$. We shall show that by removing at most $\dist\cdot m$
edges from $G$ we can obtain a graph that has arboricity at most $3\alpha$, thus reaching a contradiction
\changed{to the premise of the lemma that $G$ is $\dist$-far
 from having arboricity $3\alpha$}.
First we remove all edges in $G_\ell(\gamma)$, that is, all edges in which both endpoints belong to $A_\ell(\gamma)$.
We are left with edges that are incident to vertices in the set $V\setminus A_\ell(\gamma)$.
For each vertex $v\in V\setminus A_\ell(\gamma)$, let $a(v)$ be the number of edges it is assigned (by Algorithm~\ref{Ae}),
and recall that $a(v) \leq 3\alpha + \gamma d(v)$.
For each vertex $v$ such that $a(v) > 3\alpha$, we remove $a(v) - 3\alpha$ of the edges it is assigned (these edges can be selected arbitrarily), thus leaving it with at most $3\alpha$ assigned edges (recall that some edges may be assigned to both their endpoints).
\changed{Let $E_R$ denote the subset of edges that were removed. We have:
\begin{equation}\label{eq:removed-edges}
|E_R| \leq
|E(G_\ell(\gamma))| + \sum_{\substack{v\in V\setminus A_\ell(\gamma): \\ a(v)>3\alpha}} (a(v)-3\alpha)
 \leq
(\dist -2\gamma)m + \sum_{v\in V} \gamma d(v) = \dist \cdot m\;.
\end{equation}

 It remains to show that the set of edges that were not removed, i.e., $E\setminus E_R$,
 can be decomposed into at most $3\alpha$ forests.
 This is argued in a manner similar to~\cite{BE}. Namely, we define an acyclic orientation, where for each vertex $u$, the number of edges in $E\setminus E_R$ that are oriented from $u$ to another vertex $v$, is at most $3\alpha$ (thus defining the at most $3\alpha$ forests).  Consider an edge $(u,v)$ such that $u \in B_i$ and $v \in B_{i'}$. If $i\neq i'$, then we orient  $(u,v)$ from  the vertex that became inactive first to the vertex that became inactive second. That is, if (without loss of generality) $i<i'$, then $(u,v)$ is oriented from $u$ to $v$. If both vertices became inactive at the same iteration, that is  $i=i'$,
then we orient $(u,v)$  from the vertex with the smaller id to the vertex with the larger id.
Observe that by the definition of the orientation, for every vertex $u$, the set of edges oriented out of $u$ is a subset of the edges assigned to $u$. Furthermore, the orientation is acyclic.
Since after the removal of $E_R$, each vertex is left with at most $3\alpha$ assigned edges, we obtain an orientation as desired. Combing this with Equation~\eqref{eq:removed-edges} we reach a contradiction.
}
\end{proof}

\section{The Testing Algorithm}\label{sec:test-arb}
In this section we present and analyze our algorithm \Iba.
We assume that the distance parameter, $\eps$, is at most $1/20$
(since otherwise the algorithm can simply accept, as it is required to reject graphs that are $20\eps$-far from having arboricity at most $3\alpha$).

We start with a central procedure used by \Iba.


\subsection{Deciding whether a vertex is active}\label{sub:isact}
\label{sec:is-active}
In this subsection we present a procedure that, roughly speaking, decides whether a given vertex $v$ belongs to
the set of active vertices
$A_i(0)$ (as defined in the algorithm \AssignEdges\ from Section~\ref{sec:det-assign}). This procedure is then used
to  estimate the number of edges remaining in $G_\ell(0)$ (the subgraph induced by $A_\ell(0)$).

Observe that by the description of the algorithm \AssignEdges, for any vertex $v$, 
the decision whether $v \in A_i(0)$ can be made by considering the distance-$i$ neighborhood of $v$.
However, the size of this neighborhood may be very large, since the maximum degree in the graph is not bounded.
Hence, rather than querying for the entire distance-$i$ neighborhood, we query (in a randomized manner),
for only a small part of the neighborhood, as detailed in the procedure \IsActive.
As stated in Lemma~\ref{lem:is-active}, the procedure ensures  (with high probability), that its output is correct on $v\in A_\ell(\gamma) \subseteq A_\ell(0)$ and on $v \notin A_\ell(0)$. If $v \in A_\ell(0)\setminus A_\ell(\gamma)$, then the procedure may return any output, and as we shall see subsequently, this suffices for our purposes.

\ifnum\TALG=0
\confFigure{
			{\bf \IsActive$(v,i,\alpha, \gamma, \delta)$} \label{IsAct}
			\smallskip
			\begin{compactenum}
				\item If $d(v) \leq 3\alpha$, then return \NO.\label{Step:isa00}
				\item If $i=1$ and $d(v) > 3\alpha$, then return \YES.\label{Step:isa0}
				\item Sample a random multiset, $S_{v, i}$, of
  $t = 2\log(1/\delta)/\gamma^2$
  neighbors of $v$. \label{Step:isa1}
				\item For every $u \in S_{v ,i}$, invoke \IsActive$(u,i-1,\alpha, \gamma, \delta')$, where $\delta' = \delta/(2t)$ and let
 $\eta(v,i)$ be the fraction of vertices in $S_{v, i}$ that returned \YES. \label{Step:isa2}
				\item If $\eta(v,i) \cdot d(v) > 3\alpha + (\gamma/2)\cdot d(v)$, then return \YES, otherwise return \NO. \label{Step:isa3}
			\end{compactenum}
}
\fi

\floatname{algorithm}{Procedure}
\begin{algorithm}
 	\caption{\IsActive$(v,\ell,\alpha, \gamma, \delta)$} \label{IsAct}
	\begin{algorithmic}[1]
	\STATE Set the confidence parameter $\rho = \left(\frac{\delta \cdot \gamma}{\ell}\right)^{4\ell}$.
    \STATE Return \RecursiveIsActive$(v,\ell,\alpha, \gamma, \rho)$.
	\end{algorithmic}
 \end{algorithm}
\floatname{algorithm}{Algorithm}

\floatname{algorithm}{Procedure}
\begin{algorithm}
 	\caption{\RecursiveIsActive$(v,i,\alpha, \gamma, \rho)$} \label{RecIsAct}
 	\begin{algorithmic}[1]
	\STATE If $d(v) \leq 3\alpha$, then return \NO.\label{Step:isa00}
	\STATE If $i=1$ and $d(v) > 3\alpha$, then return \YES.\label{Step:isa0}
	\STATE Sample a random multiset, $S_{v, i}$, of
	$t = \left\lceil \frac{4\ell \log(1/\rho)}{\gamma^2} \right\rceil$
	neighbors of $v$. \label{Step:isa1}
	\STATE For every $u \in S_{v ,i}$, invoke \RecursiveIsActive$(u,i-1,\alpha, \gamma, \rho)$ 
and let
	$\eta(v,i)$ be the fraction of vertices in $S_{v, i}$ that returned \YES. \label{Step:isa2}
	\STATE If $\eta(v,i) \cdot d(v) > 3\alpha + (\gamma/2)\cdot d(v)$, then return \YES, otherwise return \NO. \label{Step:isa3}
	\end{algorithmic}
 \end{algorithm}
\floatname{algorithm}{Algorithm}


\begin{lemma}\label{lem:is-active}
For $\delta<\frac{1}{3}$ and $\gamma<1$,
the procedure \IsActive$(v,\ell,\alpha,\gamma,\delta)$ returns a value in $\{\rm \YES, \NO\}$ such that the following holds.
\begin{enumerate}
\item If  $v\notin \Ualpha$, then the procedure returns {\rm \NO} with probability at least $1-\delta$. \label{proc_peeled}
\item If $v\in A_{\ell}(\gamma)$, then the procedure returns {\rm \YES} with probability at least $1-\delta$.
\end{enumerate}
The query complexity and running time of  \IsActive$(v,\ell,\alpha,\gamma,\delta)$  are 
$O\left( \left( \frac{6\ell \cdot \log\left(\frac{\ell}{\gamma \cdot \delta }\right) }{\gamma^2} \right)^{\ell} \right)$.

\end{lemma}

\begin{proof}
For a vertex $v\in V$, consider the execution of \IsActive$(v,\ell,\alpha, \gamma, \delta)$.

For $1 \leq i \leq \ell$,
define $S_i$ to be the multiset of vertices on which \RecursiveIsActive\ is invoked with the parameter $i$.
In particular, $S_\ell = \{v\}$, and for $i < \ell$, the vertices in $S_i$ were selected in invocations
of \RecursiveIsActive\  with $i+1$.
For a vertex $u$, let
$\widehat{\eta}(u, i,\gamma)$ be the fraction of vertices in $S_{u, i}$ that 
\danachange{belong to} $A_{i-1}(\gamma)$,
and let $\widehat{\eta}(u, i,0)$ be the fraction of vertices in $S_{u, i}$ that
\danachange{belong to} $A_{i-1}(0)$ (which is a superset of $A_{i-1}(\gamma)$). 
Recall that $A_0(0) = A_0(\gamma) = V$.

For $2 \leq i \leq \ell$ we say that a vertex $u \in S_i$ is {\em $i$-successful} if one of the following holds:
\begin{enumerate}
	\item $u \in A_{i}(\gamma)$ and    $\widehat{\eta}(u, i,\gamma) \cdot d(u) > 3\alpha + (\gamma/2) d(u)$.
	\item $u \notin A_{i}(0)$ and $\widehat{\eta}(u, i,0) \cdot d(u) \leq 3\alpha + (\gamma/2) d(u)$.
	\item $u \in A_{i}(0)\setminus A_{i}(\gamma)$.
\end{enumerate}
Otherwise, it is {\em $i$-unsuccessful\/}.
For $i=1$, every vertex  is {\em $1$-successful}.

Consider a recursive call to \RecursiveIsActive$(u,i,\alpha, \gamma, \rho)$ on a vertex $u\in A_{i}(\gamma)$
for  $2 \leq i \leq \ell$.
Since $u\in A_{i}(\gamma)$, we have that $\EX[\widehat{\eta}(u,i,\gamma)\cdot d(u)] > 3\alpha+\gamma d(u)$.
By Hoeffding's inequality, the probability that $u$ is $i$-unsuccessful  is upper bounded by
\[\Pr\Big[\widehat{\eta}(u,i,\gamma) \leq \EX[\widehat{\eta}(u,i,\gamma)] -\gamma/2 \Big] < \exp(-2(\gamma/2)^2 \cdot t) \leq \rho/2\;.\]
Now consider a vertex
$u \notin A_{i}(0)$. That is, $u \in \bigcup_{i' \leq i}B_{i'}(0)$ (where $B_{i'}(\cdot)$ is as defined in the algorithm \AssignEdges).
In this case we claim that $\EX[\widehat{\eta}(u,i,0) \cdot d(u)] \leq 3\alpha$.
To verify this claim observe that since $u \in B_{i'}(0)$ (for
some $i' \leq i$), the number of neighbors that $u$ has in
\danachange{$A_{i'-1}(0)$}
is at most $3\alpha$.
The claim follows since
$A_{\danachange{i-1}}(0) \subseteq A_{\danachange{i'-1}}(0)$ (for $i'\leq i$), and by the definition
of $\widehat{\eta}(u,i,0)$ (\danachange{the fraction of vertices in $S_{u,i}$ that belong to $A_{i-1}(0)$}).
By Hoeffding's inequality, the probability that $u$ is $i$-unsuccessful  is upper bounded by
\[\Pr\Big[\widehat{\eta}(u,i,0) > \EX[\widehat{\eta}(u,i,0)] +\gamma/2 \Big] < \exp(-2(\gamma/2)^2 \cdot t) \leq \rho/2\;.\]

For  $2 \leq i \leq \ell$
we say that a vertex $u \in S_i$ is \textit{recursively $i$-successful} if $u$ is $i$-successful  and all the vertices in $S_{u, i}$ are recursively $(i-1)$-successful.
For $i=1$ every vertex is defined to be recursively $1$-successful.
By this definition, if for every $2 \leq i \leq \ell$, every $u \in S_i$ is $i$-successful, then
we also have that for every $2 \leq i \leq \ell$, every $u \in S_i$ is recursively $i$-successful.
By taking a union bound over all $2 \leq i \leq \ell$ and $u \in S_i$, we get that the probability that for some $i$ there exists $u \in S_i$ that is not $i$-successful is at most
\begin{align*}
\sum_{i=2}^\ell t^i \cdot \frac{\rho}{2}  \;& \leq\;
  t^\ell \cdot \rho  \;=\; \left\lceil \frac{4\ell \log(1/\rho)}{\gamma^2} \right\rceil^\ell \cdot \rho\\
 & \leq \; \frac{6^{\ell} \ell^{\ell}\cdot \log^{\ell}\left( \frac{\ell}{\delta \gamma}\right)}{\gamma^{2\ell}} \cdot \frac{\delta^{4\ell}\cdot \gamma^{4\ell}}{\ell^{4\ell}} < \delta,
\end{align*}
where the last inequality is by the assumption that $\delta<\frac{1}{3}$ and $\gamma<1$.

We next claim that if $u\in S_i$ is recursively $i$-successful, then the following holds:
if $u \in A_{i}(\gamma)$, then \RecursiveIsActive$(u,i,\alpha, \gamma, \rho)$ returns {\rm \YES} and  if
$u\notin A_{i}(0)$,
then \RecursiveIsActive$(u,i,\alpha, \gamma,\rho)$ returns {\rm \NO}.
We establish this claim as well by induction on $i$.

For $i=1$ (the leaves of the recursion tree), the claim follows by Steps~\ref{Step:isa00} and~\ref{Step:isa0} of the algorithm.
For the  induction step, assume the claim holds for $i-1\geq 1$, and we prove it for $i$.
If $u \in A_{i}(\gamma)$, then (since $u$ is $i$-successful), $\widehat{\eta}(u, i,\gamma) \cdot d(u) > 3\alpha + (\gamma/2) d(u)$, and since all the vertices in $S_{u, i}$ are recursively $(i-1)$-successful, then by induction, the algorithm returns {\rm \YES} in Step~\ref{Step:isa3}.
If $u\notin A_{i}(0)$,
then (again since $u$ is $i$-successful),
$\danachange{\widehat{\eta}(u,i,0)} \cdot d(u) \leq 3\alpha + (\gamma/2) d(u)$, and since all the vertices in $S_{u,i}$ are recursively $(i-1)$-successful,  by induction the algorithm returns {\rm \NO} in Step~\ref{Step:isa3}.

It remains to bound the complexity of the algorithm. Consider the recursion tree corresponding to
the complete execution of \IsActive$(v,\ell,\alpha,\gamma,\rho)$ for any vertex $v$.
Since the size of the recursion tree is $O(t^\ell) = 
 O\left( \left( \frac{6\ell \cdot \log\left(\ell/(\gamma \delta)\right) }{\gamma^2} \right)^{\ell} \right)$, 
the upper bound on the complexity of the procedure follows.
\end{proof}

\subsection{The algorithm for testing bounded arboricity}\label{subsec:test-alg}
In order to distinguish between the case that $G$ is $\eps$-close to having arboricity at most $\alpha$
and the case that it is $20\eps$-far from having arboricity at most $3\alpha$, our algorithm samples edges and applies the procedure \IsActive\ to their endpoints.
Sampling edges (almost uniformly) is done by making
use of the following theorem.

\begin{theorem}[Eden \& Rosenbaum~\cite{ER17}, rephrased]
	\label{thm:ER17} \label{Saue}
	Let $G = (V, E)$ be a graph with $n$ vertices and $m$ edges. There exists an algorithm named \Seau,
	that is given query access to $G$ and parameters $n$, $\beta$ and $\delta$.
	The algorithm returns an edge $e\in E$ with probability at least $1-\delta$,
\danachange{and conditioned on an edge being returned,}
	each edge in the graph is returned with probability in $\left[\frac{(1-\beta)}{m},\frac{(1+\beta)}{m}\right].$
	The \changed{expected} query complexity and running time of the algorithm  are
	${O}\left(\frac{n}{\sqrt{ \beta m}}\right)\changed{\cdot\log^2(n/\delta)}$.
\end{theorem}
\changed{We note that it} can be shown that the dependence on $\poly(\log n)$ in the complexity of \Seau\ can be reduced to a dependence on $\poly(\log(n/\sqrt {m}))$ (this dependence stems from estimating the average degree up to a constant factor).\footnote{This can be done by applying the algorithm of \cite{ERS16} for estimating the number of edges and slightly modifying their geometric search procedure.}

\ifnum\TALG=0
\confFigure{
			{\bf Is-bounded-arboricity$(G,\alpha,\eps)$} \label{Iba}
			\smallskip
			\begin{compactenum}
           		\item Invoke \Seau$(n,1/4,1/4)$ for $t=\settt$ times, and let $S$ be the (multi-)set of returned edges.
               \item \label{step:s}
               Let $s$ be the number of (not necessarily different) edges in $S$. If $s< \sets$, then return \YES.
        \item Set $\ell = \lceil\log_{6/5}(1/\eps)\rceil$.
		\item For each edge $(u_i, v_i) \in S$ do:
			\begin{compactenum}
				\item\label{step:isactive} Invoke \IsActive$(v_i,\ell,\alpha,\eps,\eps/2)$ and \IsActive$(u_i,\ell,\alpha,\eps,\eps/2)$.
				\item If \IsActive\ returned \YES\ on both invocations, then set $\chi_i=1$. Otherwise, set $\chi_i=0$.
			\end{compactenum}
		\item Set $\chi=\frac{1}{s}\sum_{i=1}^s \chi_i$.
		\item\label{step:output} If $\chi <10\eps $, 
		then return \YES. Otherwise, return \NO.
			\end{compactenum}
}
\else
\begin{algorithm}
\caption{\Iba$(G,\tnew{n},\alpha,\eps)$} \label{Iba}
\begin{algorithmic}[1]
	\STATE Invoke \Seau$(n,1/4,1/4)$ for $t=\settt$ times, and let $S$ be the (multi-)set of returned edges.
\STATE Let $s$ be the number of (not necessarily different) edges in $S$. If $s< \sets$, then return \YES. \label{step:s}
\STATE Set $\ell = \lceil\log_{6/5}(1/\eps)\rceil$.
\FOR{each edge $(u_i, v_i) \in S$}
	\STATE\label{step:isactive} Invoke \IsActive$(v_i,\ell,\alpha,\eps,\eps/2)$ and \IsActive$(u_i,\ell,\alpha,\eps,\eps/2)$.
\STATE If \IsActive\ returned \YES\ on both invocations, then set $\chi_i=1$. Otherwise, set $\chi_i=0$.
\ENDFOR
\STATE Set $\chi=\frac{1}{s}\sum_{i=1}^s \chi_i$.
\STATE \label{step:output} If $\chi <10\eps $, 
then return \YES. Otherwise, return \NO.
\end{algorithmic}
\end{algorithm}
\fi

\begin{theorem}\label{thm:test-arb}
\changed{Let $G$ be a graph over $n$ vertices and $m$ edges.}
If $G$ is $\eps$-close to having arboricity at most  $\alpha$, then {\rm\bf Is-bounded-arboricity}
returns {\rm \YES} with probability at least $2/3$, and if $G$ is
$20\eps$-far
from having arboricity at most $3\alpha$, then {\rm \textbf{Is-bounded-arboricity}} returns {\rm \NO} with probability at least $2/3$.

The query complexity and running time of \Iba\ are
$$\comp$$
in expectation.  
\end{theorem}

\begin{proof}
By Theorem~\ref{Saue}, each invocation of \Seau$(n,1/4,1/4)$ succeeds with probability at least $3/4$. By the multiplicative Chernoff bound and by the setting of $t=\settt$, it follows that with probability at least
$5/6$,
at least $1/2$ of the invocations return an edge. Hence, $s\geq \sets$ with probability at least
$5/6$ and the algorithm continues to the following steps. We henceforth condition on this event.
	
	We say that the procedure \IsActive\ is \emph{correct} when invoked with a vertex $v$
in Step~\ref{step:isactive} of the algorithm
if $v\in \Umax$ and \IsActive\ returns \YES, or if $v\notin \Ualpha$ and \IsActive\ returns \NO.
For a subgraph $G'$ of $G$ we let $m(G')$ denote the number of edges in $G'$.

\medskip	
We first consider the case that $G$ is $\eps$-close to having arboricity at most $\alpha$.
By Lemma~\ref{lem:det-assign-ub}, in this case $m(\rG) \leq 5\eps m$.
For each $i\in [s]$ such that the edge $(u_i, v_i)$ does not belong to $\rG$, it holds that either $u_i$ or $v_i$ is not in $\Ualpha$. Hence, by Lemma~\ref{lem:is-active}, \IsActive\ returns \YES\ on both vertices with probability at most $\eps/2$ (recall that \IsActive\ is called in Step~\ref{step:isactive} with the confidence parameter
$\delta$ set to $\eps/2$).
	For each edge $(u_i, v_i) \in \rG$, we upper bound the probability that \IsActive\ returns \YES\ on both vertices by $1$. By Theorem~\ref{Saue}, when \Seau\ is invoked with parameters $\beta = 1/4$ and $\delta=1/4$, if the algorithm returns an edge, then each edge in the graph is returned with probability in $[(3/4)/m,(5/4)/m]$. Therefore, it holds that
	\[\EX[\chi_i] \leq \frac{\eps}{2} \cdot \frac{(5/4)\cdot m}{m} + \frac{(5/4)\cdot m(\rG)}{m}\leq \frac{5\cdot m(\rG)/4+\eps m}{m}\;.\]
Since $m(\rG)\leq 5\eps m$, we get that
$\EX[\chi_i] \leq 8\eps$ for every $i \in [s]$.

	By the multiplicative Chernoff bound and since  $s\geq\sets$,
	\begin{align*}  \Pr\left[ \frac{1}{s}\sum_{i=1}^s \chi_i > \left(1+\frac{1}{20}\right)\cdot 8\eps \right]
	 \ifnum\soda=1
	\\ <
	\else
	<
	\fi
	 \exp\left(-\frac{(1/20)^2\cdot 8\eps \cdot s}{3}\right) < 1/6 \;. 
	\end{align*}
	It follows that if $G$ is $\eps$-close to having arboricity at most $\alpha$, then
either the algorithm returns \YES\ in Step~\ref{step:s}, or
with probability at least
$5/6$,
 $\chi\leq 9\eps $, which causes the algorithm to return \YES\ in Step~\ref{step:output}.

\medskip	
Now consider the case that $G$ is  $20\eps$-far from having arboricity at most $3\alpha$.
By Lemma~\ref{lem:det-assign-lb}, setting $\eps' = 20\eps$ and $\gamma=\eps$,
in this case we get that
$\danachange{m(\Gmax)}
    > 18\eps m$.	

For each $i\in [s]$ such that the edge
$(u_i, v_i)$ belongs to $\Gmax$, we get that $\chi_i=1$ if the invocations of \IsActive\ on both $u_i$ and $v_i$
return \YES. Since for $(u_i,v_i)$ in $\Gmax$ both
 $u_i$ and $v_i$ belong to $\Umax$, by Lemma~\ref{lem:is-active} and the union bound, \IsActive\ returns \YES\ on both vertices with probability at least $1-\eps$. Hence \changed{(using our assumption that $\eps \leq 1/20$)},
 \[\EX[\chi_i]\geq \frac{(1-\eps)\cdot (3/4)\cdot m(\Gmax)}{m}\geq 12\eps .\]

	By the multiplicative Chernoff bound and since $s \geq \sets$,
	\begin{align*}
\Pr\left[ \frac{1}{s}\sum_{i=1}^s \chi_i < \left(1-\frac{1}{20}\right)\cdot {12\eps} \right]
< \exp\left(-\frac{(1/20)^2\cdot 12 \eps \cdot s}{2}\right) < 1/6.
	\end{align*}
It follow that (conditioned on $s \geq \sets$),
with probability at least $5/6$,
	$\chi\geq 10\eps$, which causes the algorithm to return \NO\
in Step~\ref{step:output}. Since the probability that $s<\sets$ is at most $1/6$,
if $G$ is $20\eps$-far from having arboricity at most $3\alpha$, then the algorithm returns \NO\ with probability at least
\danachange{$2/3$}.

\medskip
It remains to bound the complexity of the algorithm.	
	By Theorem~\ref{Saue}, the $t=\settt$ invocations of \Seau\ with parameters $\beta = 1/4$ and $\delta=1/4$  take
${O}\left(\frac{n\tnew{\cdot \log^2 n}}{\eps\sqrt m}\right)$ time.
	In each step of the for loop there are at most two invocations of the procedure \IsActive\ with parameters $\gamma=\eps$ and $\delta=\eps/2$. By Lemma~\ref{lem:is-active}, the query complexity and running time resulting from each of these invocations are
$O\left(\left(\frac{6\ell\log(\ell/\eps^2)}{\eps^2}\right)^\ell\right)$.
Since $\ell=\lceil \log_{6/5}(1/\eps)\rceil $,
 the total query complexity and running time are
$\comp$, as claimed.
\end{proof}

\strikeout{
\subsection{Estimating the number of edges between high-degree vertices}\label{subsec:high}

In this subsection we provide a procedure for distinguishing between the case that the number of edges between
vertices whose degree is above a certain threshold is relatively small, and the case in which the number of such edges is relatively large.

 \begin{definition}\label{def:high-deg}
 	We say that a vertex $v$ is a \emph{high-degree} vertex if $d(v)>2\alpha/\eps$. We let $H$ denote the set of high-degree vertices in the graph.
 If both endpoints of an edge are high-degree vertices, then we say that the edge is a \emph{high} edge. Otherwise, if either endpoint of the edge is low, we say that the edge is  a \emph{low} edge.
 \end{definition}

%
 \begin{claim}\label{clm:high-deg-edges}
 	If $G$ is $\eps$-close to having arboricity at most $\alpha$, then the
 number of high edges in $G$ is at most $2\eps m$. 
 \end{claim}

\ifnum\conf=0
 \begin{proof} 	
 	Let $E(H)$ denote the set of high edges.
 	By the definition of high-degree vertices, $|H|<\eps m/\alpha$.
 Since $G$ is $\eps$-close to having arboricity at most $\alpha$, $|E(H)|\leq \alpha|H| +\eps m$, implying that $|E(H)| \leq 2\eps m$.
 \end{proof}	
 \fi

\ifnum\conf=1
The proof of the above claim as well as the next lemma and the procedure it refers to are deferred to the appendix (see Appendix~\ref{sec:appB}).

\else
 Our procedure for estimating the number of high edges makes use of the following theorem.
 \begin{theorem}[Eden \& Rosenbaum~\cite{ER17}, rephrased]
	\label{thm:ER17} \label{Saue}
	Let $G = (V, E)$ be a graph with $n$ vertices and $m$ edges. There exists an algorithm named \Seau,
 that is given query access to $G$ and parameters $n$, $\eps$ and $\delta$.
The algorithm returns an edge $e\in E$ with probability at least $1-\delta$, where
each edge in the graph is returned with probability in $\left[\frac{(1-\eps)}{m},\frac{(1+\eps)}{m}\right].$
The query complexity and running time of the algorithm  are
   $\tilde{O}\left(\frac{n }{\sqrt{ \eps m}}\right)\cdot\poly(\log(1/\delta))$.
\end{theorem}

\confFigure{
			{\bf Estimate-high-edges$(n,\alpha,\eps,\delta, \bar{m})$} \label{Ehe}
			\smallskip
			\begin{compactenum}
				\item Set $r=\setr$ and 	$\delta'=\delta/(2r)$.
				\item For $i=1$ to $r$ do:
				\begin{compactenum}
				\item Invoke \Seau$(n,0.1, \delta')$
and let $(u_i,v_i)$ be the returned edge. If no edge was returned then return \FEW.
				\item Set $\chi_i=1$ if both $u_i$ and $v_i$ are high-degree vertices, and set $\chi_i=0$ otherwise.
				\end{compactenum}
				\item Set $\chi=\frac{1}{r}\sum_{i \in Y} \chi_i$.
				\item If  $\chi > 2.6\eps$,    
                         then return \MANY. Otherwise, return \FEW.
			\end{compactenum}
}
\fi

\begin{lemma}\label{lem:Ehe}
\ifnum\conf=1
There exists a procedure \Ehd\  that
\else
The procedure \Ehd\
\fi
returns a value in {\rm \{\MANY,\FEW\}} and satisfies the following.
\begin{enumerate}
\item If the graph $G$ has more than 
$4\eps m$
high edges, then the procedure returns {\rm \MANY} with probability at least $1-\delta$.
\item If the graph $G$ has at most
$2\eps m$
high edges, then the procedure returns {\rm \FEW} with probability at least $1-\delta.$
\end{enumerate}
\sloppy
The expected query complexity and running time of the procedure are
$\tilde{O}\left(\frac{n}{\sqrt {m}} \cdot \frac{\log(1/\eps)}{\eps} \right)\cdot \poly(\log(1/\delta))$.
\end{lemma}
\ifnum\conf=0

\begin{proof}
By Theorem~\ref{thm:ER17}, the invocation of \Seau$(n,0.1,\delta')$ returns an edge $(u_i,v_i)$ with probability at least $1-\delta'$. Therefore, by the setting of $\delta'$ and by the  union bound, with probability at least $\delta/2$ all invocations will return an edge. We henceforth condition on this event.

Also by Theorem~\ref{thm:ER17}, for every edge in the graph $G$, the probability that it will be the returned edge is in $\left[0.9\cdot \frac{1}{m},1.1\cdot \frac{1}{m}\right]$. 
Hence, if $G$ has more than $4\eps m$ 
high edges, then
$\EX[\chi]\geq 0.9 \cdot 4\eps > 3\eps$.
By the multiplicative Chernoff bound, and by the setting of $r$,
\[ \Pr\left[ \frac{1}{r}\sum_{i}^r \chi_i \leq 2.6\eps\right]
        < \exp\left(-\frac{0.1^2\cdot 3\eps\cdot r}{3}\right)
           <\frac{\delta}{2}.
\]
It follows that in this case the procedure will return \MANY\ with probability at least $1-\delta$.

We now turn to the case that that $G$ has at most $2\eps m$ 
high edges, so that $\EX[\chi]\leq   2.2 \eps$.
By the multiplicative Chernoff bound, and by the setting of $r$,
\[ \Pr\left[ \frac{1}{r}\sum_{i}^r \chi_i > 2.6\eps\right] < \exp\left(-\frac{0.1^2\cdot 2.2\eps\cdot r}{2}\right)
<\frac{\delta}{2}.
\]
Therefore, the procedure will return \FEW\ with probability at least $1-\delta$.

\sloppy
It remains to prove the third item of the lemma. By Theorem~\ref{thm:ER17}, the expected query complexity and running time of each invocation of \Seau\ with parameters $0.1$ and $\delta'$ are
$\tilde{O}\left(\frac{n}{\sqrt {m}} \cdot \poly\log(1/\delta')\right)$.
Therefore, by the setting of $r$ and $\delta'$, the expected query complexity and running time are
$\tilde{O}\left(\frac{n}{\sqrt {m}} \cdot \frac{\log(1/\eps)}{\eps} \right)\cdot \poly(\log(1/\delta))$.
\end{proof}
\fi

}

\section{Lower bounds} \label{sec:lower-bound}

The following lower bounds are quite simple and are proved here for the sake of completeness.

\begin{claim}\label{clm:lb1}
For a graph $G$ let $n$ denote the number of vertices in $G$ and let $\om$ be a constant factor approximation of the  number of edges, $m$, in $G$. 
Let ${\mathcal A}$ be an algorithm that is given query access to a graph $G$ as well as parameters $n$, $\om$, $\eps<1/100$, and $\alpha<\sqrt{\eps \om}$. The algorithm $\mathcal{A}$ is required to distinguish
with probability at least $2/3$ between the case that $G$ has arboricity at most $\alpha$ and the case that $G$ is  $20\eps$-far from having arboricity at most $3\alpha$. Then ${\mathcal A}$ must perform $\Omega\left(\frac{n}{\sqrt{\epsilon  m}} \right)$ queries.
\end{claim}
\begin{proof}
Consider the following two families of graphs $\mG_1$ and $\mG_2$. Every graph in $\mG_1$ consists of
three disjoint subgraphs $G^1_1, G^1_2$ and $G^1_3$ as described next.
$G^1_1$ is an independent set of size $n-2\om/\alpha-2\sqrt{100 \eps \om}$;
$G^1_2$ is a bipartite graph with $\om/\alpha$ vertices on each side and  $\alpha$ perfect matchings between the two sides;
$G^1_3$ is an independent set of size $2\sqrt{100 \eps \om}$;
The graphs within the family differ from one another only by 
the labeling of the vertices.
The graphs in the second family $\mG_2$ also consists of three disjoint subgraphs $G^2_1, G^2_2, G^2_3$. Here we have $G^2_1=G^1_1, G^2_2=G^1_2$ and $G^2_3$ is a complete bipartite graph with $\sqrt{100 \eps \om}$ vertices on each side.
As before, the graphs within the family differ only by 
the labeling of the vertices.

All graphs in $\mG_1$ have $\om$ edges, and all graphs in $\mG_2$ have $\om + 100\eps\om < 2\om$ edges.
Furthermore, all graphs in $\mG_1$ have arboricity $\alpha$, while in $\mG_2$ all graphs are $20 \eps$-far from having arboricity at most $3\alpha$.
To verify the latter claim, observe that after removing $20\eps m < 40\eps \om$ edges from any graph in $\mG_2$, the number of edges remaining in the subgraph $G^2_3$ is greater than $100\eps \om - 40 \eps\om = 60\eps\om$. Since the number of vertices in $G^2_3$ is $2\sqrt{100\eps\om}=20\sqrt{\eps \om}$, the arboricity of $G^2_3$ (after the removal of the aforementioned edges) is greater than $3\sqrt{\eps\om} > 3\alpha$.

In order to distinguish between a graph drawn uniformly from the first family and a graph drawn uniformly from the second family, an algorithm must witness a vertex in $G^i_3$ for $i \in  \{1,2\}$.
Since the probability of witnessing such a vertex for both $i=1$ and $i=2$ is 
$\frac{|V(G^i_3)|}{n} = \frac{2\sqrt{100 \eps \om}}{n}$, at least $\Omega\left(\frac{n}{\sqrt{\epsilon m}}\right)$ queries are required in order to distinguish between the two families  with 
probability at least $2/3$.
\end{proof}

\tnew{
The proof of Claim~\ref{clm:lb1} relies on the ability to construct a lower bound instance where we ``hide'' a small set of vertices with very high density. When the algorithm is also given the exact number of edges in the graph this is no longer possible, and the above lower bound does not hold.
Instead, for the case where $m$ is known, we prove a weaker lower bound of $\Omega(\frac{n\alpha}{\tnew{\eps} m})$.
}

\begin{claim}\label{clm:lb2}
For a graph $G$ let $n$ denote the  number of vertices in $G$ and let $m$ denote the number of edges.
Let ${\mathcal A}$ be an algorithm that is given query access to a graph $G$ as well as parameters $n,m,\eps<1/100$ and $\alpha \leq 
    \sqrt{\eps m}$.
The algorithm $\mathcal A$ is required to distinguish with 
probability at least $2/3$ between the case that $G$ has arboricity at most $\alpha$ and the case that $G$ is  $20\eps$-far from having arboricity at most $3\alpha$.
Then ${\mathcal A}$ must perform $\Omega\left(\frac{n\alpha}{ \tnew{\eps}m} \right)$ queries.
\end{claim}
\begin{proof}

We describe two families of graphs $\mG_1$ and $\mG_2$. Each graph in the first family $\mG_1$ consists of four disjoint subgraphs, $G^1_1, G^1_2, G^1_3$ and $G^1_4$, which are defined as follows: $G^1_1$ is an independent set over $n-2m/\alpha -2\sqrt{100\eps m}$ vertices; $G^1_2$ is a bipartite graph with $(1-100\eps)m/\alpha$ vertices on each side, and $\alpha$ perfect matchings between the sides; $G^1_3$ is a bipartite graph  with $100 \eps m/\alpha$ vertices on each side, and $\alpha$ perfect matchings between the sides; $G^1_4$ is an independent set over $2\sqrt{100\eps m}$ vertices.
The graphs in the second family $\mG_2$ also consists of four disjoint subgraphs $G^2_1, G^2_2, G^2_3$ and $G^2_4$. Here we have $G^2_1 = G^1_1$ and $G^2_2 = G^1_2$, $G^2_3$ is an independent set over $200\eps m/\alpha$ vertices and $G^2_4$ is a complete bipartite graph with $\sqrt{100\eps m}$ vertices on  each side. Within each family, the graphs differ only by the labeling of the vertices. By the above description we get that graphs in both families have exactly $m$ edges. Furthermore, the graphs in $\mG_1$ have arboricity $\alpha$, and the graphs in $\mG_2$ are $20\eps$-far from having arboricity $3\alpha$ (this follows similarly to what was shown in the proof of Claim~\ref{clm:lb1}).

We assume without loss of generality that every neighbor query $(u,i)$ or pair query $(u,v)$ is preceded by one or two degree queries $d(u)$ or $d(u), d(v)$, respectively. Furthermore,
assume that whenever the algorithm queries for the degree of some vertex $u$, it also gets an index $j \in \{1,2,3,4\}$ indicating to which of the subgraphs $G^{i}_1, G^{i}_2, G^{i}_3, G^{i}_4$ it belongs (without revealing the value of $i$).  Since $G^1_1=G^2_1$ and $G^1_2=G^2_2$, it is clear that in order to distinguish between a graph drawn from $\mG_1$ to a graph drawn from $\mG_2$ any algorithm must hit either $G^{i}_3$ or $G^{i}_4$ for the corresponding $i$ value. Since for both $i=1$ and $i=2$,  $|G^{i}_3|+|G^{i}_4| = O(\frac{\eps m}{\alpha})$ (recall that we assume that $\alpha\leq \sqrt{\eps m}$), we have that hitting either $G^{i}_3$ or $G^{i}_4$ occurs with probability $O\left(\frac{\eps m}{n \alpha}\right)$, so that  $\Omega\left(\frac{n\alpha}{\eps m}\right)$ queries are required in order to distinguish the two families with 
probability at least $2/3$.
\end{proof}

Finally we establish that there is no  one-sided error algorithm for bounded arboricity  that performs a  number of queries
that is sublinear in $n$.
\begin{claim}\label{clm:one-sided}
Let ${\mathcal A}$ be an algorithm that is given query access to a graph $G$ as well as parameters
\danachange{$n$, $\alpha \geq 2$
and $\eps\leq \frac{1}{4}$ (where $n$ is the number of vertices in $G$).}
It is required to accept $G$ with probability 1 if $G$ has arboricity at most $\alpha$ and reject  $G$ with probability at least $2/3$ if
 $G$ is  $\eps$-far from having arboricity at most $3\alpha$.
 Then ${\mathcal A}$ must perform $\Omega(n)$ queries.
\end{claim}

\begin{proof}
\changed{
We shall prove that for any $\alpha \geq 2$, sufficiently large $n$ and $\eps \leq 1/4$, there exists a graph $G$ over $n$ vertices for which the following two conditions hold.
 On one hand, $G$ is $\eps$-far from having arboricity $3\alpha$. On the other hand, for any $k \leq n/c$, where $c$ is a sufficiently large constant, every induced subgraph of $G$ over $k$ vertices has arboricity at most $\alpha$. Therefore, for this graph, any one-sided error algorithm must perform $\Omega(n)$ queries.

 Consider selecting $G$ according to the distribution $G(n,p)$ where $p = \frac{10\alpha}{n}$.
 That is, for each pair of vertices, the probability that we have an edge between these two vertices is $p$ (and the corresponding events for different pairs of vertices are independent).
 The expected number of edges in $G$ is ${n\choose 2}\cdot p = 5\alpha(n-1)$. By applying the multiplicative Chernoff bound, with very high probability, the number of edges in $G$ is at least $4\alpha(n-1)$ (i.e., at least $4/5$ of the expected value), so that by Equation~(\ref{eq:Nash-Williams}), $G$ is at least $1/4$-far from having arboricity $3\alpha$.

Let $k$ be an integer such that $4< k \leq n/c$, and let $K$ be a subset of $k$ vertices.
We next upper bound the probability that the number of edges in the subgraph induced by $K$, denoted $m(K)$, is more than $\alpha(k-1)$. For any fixed set $B$ of $\alpha(k-1)$ pairs of vertices, the probability that we get an edge between every pair in $B$ is $p^{\alpha(k-1)}$. Taking a union bound over all such subsets $B$, and using the inequality ${y \choose x} \leq \left(\frac{e\cdot y}{x}\right)^x$, the probability that $m(K) > \alpha(k-1)$ is upper bounded by
\[
{{k\choose 2}\choose \alpha(k-1)} \cdot p^{\alpha(k-1)} \leq \left( \frac{e\cdot k \cdot p}{2\alpha}\right)^{\alpha(k-1)}
   =   \left( \frac{5e\cdot k}{n}\right)^{\alpha(k-1)}\;.
\]
By taking a union bound over all ${n\choose k}$ subsets of size $k$ we get that the probability that there exists any such subset of size $k$ is upper bounded by
\begin{eqnarray*}
{n\choose k} \cdot \left( \frac{5e\cdot k}{n}\right)^{\alpha(k-1)} &\leq& \left(\frac{e\cdot n}{k}\right)^k
            \cdot  \left( \frac{5e\cdot k}{n}\right)^{\alpha(k-1)} \;\leq\; (5e^2)^{\alpha(k-1)} \cdot
            \left(\frac{k}{n}\right)^{\alpha(k-1)-k} \\
            &\leq& \left( (5e^2)^3 \cdot \frac{k}{n}  \right)^{\alpha(k-1)-k}
             \; \leq \; \left( (5e^2)^3 \cdot \frac{k}{n}  \right)^{k-2}\;,
\end{eqnarray*}
where we have used the fact that $\frac{\alpha(k-1)}{\alpha(k-1)-k} \leq 3$ for $\alpha \geq 2$ and $k \geq 4$,
and that $\alpha(k-1)-k \geq k-2$ for $\alpha \geq 2$. By setting $c = 2(5e^2)^3$, this probability is upper bounded by $2^{-(k-2)}$, and by summing over all $k>4$ we get that the probability is bounded away from 1.
}
\end{proof}

As mentioned in the introduction, for the case of $\alpha=1$ (cycle-freeness),
there is a one-sided error testing algorithm~\cite{CGRSSS} that performs  $\tilde{O}(\sqrt n)$ queries (and these many queries are necessary~\cite{GR02}).
	\input{variations}

\newcommand{\ta}{\widetilde{\alpha}}
\newcommand{\oa}{\overline{\alpha}}
\newcommand{\as}{{\alpha^*}}
\newcommand{\Ea}{\hyperref[Erd]{\color{black} \bf Estimate-Corrected-Arboricity}}

\section{Estimating the corrected arboricity} \label{sec:corrected}

	In this section we present a procedure for estimating what we refer to as the \emph{$\eps$-corrected-arboricity} of a graph $G$. The $\eps$-corrected-arboricity of a graph $G$ is the minimal arboricity of a graph that $G$ can be ``corrected into". That is, 	
	it is the minimal arboricity over all the graphs that are $\eps$-close to $G$ (See Definition~\ref{def:corrected}).
	\changed{
		The procedure performs a standard geometric search on the value of the ``corrected arboricity'' using the testing algorithm \Iba, and we 
provide the procedure
here for the sake of completeness.
		More precisely, the procedure first obtains an estimate  $\om$ of $m$. Then it starts with a guess value $\ta=1$, and for each guess value $\ta$ it invokes \Iba\ with $\ta$ for $O(\log \log {\om})$ times.  If the majority of  votes return \YES, then the algorithm returns $\ta$, and otherwise it continues with $\ta=2\ta$.  The search ends when $\ta$ exceeds $\sqrt{\om}$, at which point, the algorithm simply returns $\sqrt{\om}$. (Recall that for every graph $G$, $\alpha \leq \sqrt m$.)
	 }
	We note that the overhead of this procedure, compared to the testing algorithm, is a factor of $\poly(\log n)$.
	
	\begin{theorem}[Goldreich \& Ron~\cite{GR08}, rephrased] \label{thm:GR08} There exists a procedure that when invoked with a graph $G$ and a confidence parameter $\delta$, returns a value $\om$ such that with probability at least $1-\delta$, $\om \in [m,2m]$, where $m$ is the number of edges in $G$. The expected running time of the procedure is ${O}(n\log(\tnew{n}/\delta)/\sqrt m)$.
	\end{theorem}	

\ifnum\TALG=0		
	\begin{figure}[htb!]
		\fbox{
			\begin{minipage}{0.9\textwidth}
				{\bf \Ea$(G,n,\eps)$} \label{Ea}
				\smallskip
				\begin{compactenum}
				\item Invoke \cite{GR08}$(G,1/9)$ to obtain an estimate $\om$ of the number of edges in $G$.
				\item Let $\ta=1$ and $\delta=1/(9\log \om).$
				\item While $\ta\leq \sqrt {\om}$ do: \label{step:EaWhile}
					\begin{compactenum}
					\item Invoke \Iba$(G,n,\ta,\eps)$ for $10\log(1/\delta)$ times. \label{step:EaInvoke}
					\item If more than half of the invocations returned \textsf{Yes}, then return $\oa=\ta$.\label{step:EaReturn}
					\item Let $\ta=2\ta$. \label{step:EaDouble}							
					\end{compactenum}
				\item Return $\sqrt {\om}$.
				\end{compactenum}
			\end{minipage}
		}
	\end{figure}
\else
\floatname{algorithm}{Algorithm}
\begin{algorithm}
\caption{\Ea$(G,n,\eps)$} \label{Ea}
\begin{algorithmic}[1]
	\STATE Invoke \cite{GR08} with parameters $G$ and $1/9$ to obtain an estimate $\om$ of the number of edges in $G$. \label{step:est-m}
	\STATE Let $\ta=1$ and $\delta=1/(9\log \om).$
	\WHILE {$\ta\leq \sqrt {\om}$}  \label{step:EaWhile}
		\STATE Invoke \Iba$(G,n,\ta,\eps)$ for $10\log(1/\delta)$ times. \label{step:EaInvoke}
		\STATE If more than half of the invocations returned \textsf{Yes}, then return $\oa=\ta$.\label{step:EaReturn}
		\STATE Let $\ta=2\ta$. \label{step:EaDouble}							
	\ENDWHILE
	\RETURN $\sqrt {\om}$.
\end{algorithmic}
\end{algorithm}
\floatname{algorithm}{Algorithm}
\fi

	\begin{definition}[$\eps$-Corrected Arboricity]\label{def:rStar} \label{def:corrected}
		Let $\mathcal{G}_n$ denote the family of graphs over $n$ vertices.
		For a graph $G\in \mG_n$ let $\as(G,\eps)$ denote the minimal value $\alpha'$ such that $G$ is $\eps$-close to a graph $G'$ with arboricity at most $\alpha'$. That is, $\as(G,\eps)=\min_{G' \in \mG_n}\{ \alpha(G') \mid dist(G,G')\leq \eps \cdot m(G)\}$.
		We refer to the value $\as$ as the \emph{$\eps$-corrected arboricity} of the graph $G$.
		\end{definition}
	\begin{claim} The algorithm \Ea\ when invoked with a graph $G$ over $n$ vertices and parameter $\eps$ returns a value $\oa$ such that with probability at least $2/3$,
	$\as(G,20\eps)/3 \leq \oa \leq 2\as(G,\eps)$.
	 The expected query complexity and running time of the algorithm are $$\tilde{O}\left(\frac{n}{\eps \sqrt m}+ \left(\frac{1}{\eps}\right)^{O(\log (1/\eps))}\right)\;.$$
		\end{claim}
	\begin{proof}
	By Theorem~\ref{thm:GR08}, with probability at least $8/9$, the value $\om$ is such that $\om\in[m,2m]$. We henceforth condition on this event.
	
	For every value $\ta$ such that $\ta< \as(G,20\eps)/3$, it holds that
	$G$ is at least $20\eps$-far from every graph with arboricity at most $3\ta$. Hence, by Theorem~\ref{thm:test-arb}, every invocation of \Iba$(G,n,\ta,\eps)$  returns \textsf{No} with probability at least $2/3$, and by a simple amplification argument, the probability that more than half of the invocations
	return  \YES\ in Step~\ref{step:EaReturn} is at most $\delta$.
	Therefore, with probability at least $1-\delta$, \Ea\ will continue to run with a value $2\ta$.
	Since there are at most $\log(\sqrt m)$ invocations with a value $\ta < \as(G,20\eps)/3$, by a union bound, the probability that \Ea\ will return a value $\ta<\as(G,20\eps)/3$ is at most $1/9$.
	
	Once we reach a value $\ta$ such that
	$\ta \geq \as(G,\eps)$, then by the definition of $\as$, $G$ is $\eps$-close to having arboricity at most $\ta$, and therefore, by Theorem~\ref{thm:test-arb}, every invocation of \Iba$(G,\ta,\eps)$ returns \textsf{Yes} with probability at least $2/3$. Hence, with probability at least $1-\delta$, more than half of the invocations of \Iba\  return \YES\, and the algorithm returns $\ta$. Since we increase $\ta$ by a factor $2$ at every step it holds that we will reach a value $\ta$ such that $\ta\in[\as(G,\eps),2\as(G,\eps)]$, and by the above, once we reach such a value \Ea\ will return $\ta$ with probability at least $1-\delta>8/9$.
	
	By a union bound, with probability at least $2/3$, \Ea\ returns a value $\oa$ such that $\as(G,20\eps)/3 \leq \oa\leq 2\as(G,\eps)$.
	
	\tnew{By Theorem~\ref{thm:GR08}, estimating the number of edges in Step~\ref{step:est-m}, takes $O\left(\frac{n \cdot \log^2 n}{\sqrt m}\right)$ time in expectation.}
	By Theorem~\ref{thm:test-arb}, every invocation of the while loop in Step~\ref{step:EaWhile} takes
	$\comp$
	 time in expectation.
	Since the number of iterations is at most $\log(\bar m)$ (and $\bar m\leq n^2$), the query complexity and running time are
	 $$\tilde{O}\left(\frac{n}{\eps \sqrt m}+ \left(\frac{1}{\eps}\right)^{O(\log (1/\eps))}\right)$$
		  in expectation.
	
	\end{proof}


\subsection*{Acknowledgments}
We would like to thank the reviewers of the ACM Transactions on Algorithms journal for their helpful comments.

\bibliographystyle{plain}
\bibliography{arb_bib}

\appendix
\ifnum\conf=1
We make use of Hoeffding's inequality \cite{hoeffding}, stated next.
For $i=1, \ldots, s$, let $\chi_i$ be a $0/1$ values random variable, such that $\Pr[\chi_i=1] = \mu$.
Then for any $\gamma \in (0,1]$,
\[ \Pr\left[\frac{1}{s}\sum\limits_{i=1}^s \chi_i > \mu+\gamma \right] < \exp\left( -2\gamma^2 s\right) \]
and
\[	\Pr\left[\frac{1}{s}\sum\limits_{i=1}^s \chi_i < \mu-\gamma \right] < \exp\left( -2\gamma^2 s\right)\;.\]
We also make use of the following version of  the multiplicative Chernoff bound \cite{chernoff}.
For $i=1, \ldots, s$, let $\chi_i$ be a random variables taking values in $[0,B]$, such that  $\EX[\chi_i]=\mu$. Then
for any $\gamma \in (0,1]$,
\[ \Pr\left[\frac{1}{s}\sum\limits_{i=1}^m \chi_i > (1+ \gamma)\mu \right] < \exp\left( -\frac{\gamma^2 \mu s}{3B}\right)
\]
and
\[\Pr\left[\frac{1}{s}\sum\limits_{i=1}^m \chi_i < (1- \gamma)\mu \right] < \exp\left( -\frac{\gamma^2 \mu s}{2B}\right)\;.\]
\fi

\ifnum\conf=1

\section{Missing details for Subsection~\ref{subsec:high}}\label{sec:appB}

\begin{proofof}{Claim~\ref{clm:high-deg-edges}}
	Let $E(H)$ denote the set of high edges.
	By the definition of high-degree vertices, $|H|<\eps m/\alpha$.
	Since $G$ is $\eps$-close to having arboricity at most $\alpha$, $|E(H)|\leq \alpha|H| +\eps m$, implying that $|E(H)| \leq 2\eps m$.
\end{proofof}	

\smallskip
Our procedure for estimating the number of high edges makes use of the following theorem.
\begin{theorem}[Eden and Rosenbaum~\cite{ER17}
	\ifnum\soda=0, rephrased
	\fi
	]
	\label{thm:ER17} \label{Saue}
	Let $G = (V, E)$ be a graph with $n$ vertices and $m$ edges. There exists an algorithm named \Seau,
	that is given query access to $G$ and parameters $n$, $\eps$ and $\delta$.
	The algorithm returns an edge $e\in E$ with probability at least $1-\delta$, where
	each edge in the graph is returned with probability in $\left[\frac{(1-\eps)}{m},\frac{(1+\eps)}{m}\right].$
	The query complexity and running time of the algorithm  are
	$\tilde{O}\left(\frac{n }{\sqrt{ \eps m}}\right)\cdot\poly(\log(1/\delta))$.
\end{theorem}

\begin{figure}[htb!]
	\fbox{
		\begin{minipage}
			{0.465\textwidth}
			\begin{raggedright}
			
			{\bf Estimate-high-edges$(n,\alpha,\eps,\delta, \bar{m})$} \label{Ehe}
			\smallskip
			\begin{compactenum}
				\item Set $r=\setr$ 
				\item For $i=1$ to $r$ do:
				\begin{compactenum}
					\item Invoke \Seau$(n,0.1, \delta')$
					and let $(u_i,v_i)$ be the returned edge. If no edge was returned then return \FEW.
					\item Set $\chi_i=1$ if both $u_i$ and $v_i$ are high-degree vertices, and set $\chi_i=0$ otherwise.
				\end{compactenum}
				\item Set $\chi=\frac{1}{r}\sum_{i \in Y} \chi_i$.
				\item If  $\chi > 2.6\eps$, 
				then return \MANY. Otherwise, return \FEW.
			\end{compactenum}
			\end{raggedright}
		\end{minipage}
	}
\end{figure}

\begin{proofof}{Lemma~\ref{lem:Ehe}}
	By Theorem~\ref{thm:ER17}, the invocation of \Seau$(n,0.1,\delta')$ returns an edge $(u_i,v_i)$ with probability at least $1-\delta'$. Therefore, by the setting of $\delta'$ and by the  union bound, with probability at least $\delta/2$ all invocations will return an edge. We henceforth condition on this event.
	
	Also by Theorem~\ref{thm:ER17}, for every edge in the graph $G$, the probability that it will be the returned edge is in $\left[0.9\cdot \frac{1}{m},1.1\cdot \frac{1}{m}\right]$. 
	Hence, if $G$ has more than $4\eps m$ 
	high edges, then
	$\EX[\chi]\geq 0.9 \cdot 4\eps > 3\eps$.
	By the multiplicative Chernoff bound, and by the setting of $r$,
	\[ \Pr\left[ \frac{1}{r}\sum_{i}^r \chi_i \leq 2.6\eps\right]
	< \exp\left(-\frac{0.1^2\cdot 3\eps\cdot r}{3}\right)
	<\frac{\delta}{2}.
	\]
	It follows that in this case the procedure will return \MANY\ with probability at least $1-\delta$.
	
	We now turn to the case that that $G$ has at most $2\eps m$ 
	high edges, so that $\EX[\chi]\leq   2.2 \eps$.
	By the multiplicative Chernoff bound, and by the setting of $r$,
	\[ \Pr\left[ \frac{1}{r}\sum_{i}^r \chi_i > 2.6\eps\right] < \exp\left(-\frac{0.1^2\cdot 2.2\eps\cdot r}{2}\right)
	<\frac{\delta}{2}.
	\]
	Therefore, the procedure will return \FEW\ with probability at least $1-\delta$.
	
	\sloppy
	It remains to prove the third item of the lemma. By Theorem~\ref{thm:ER17}, the expected query complexity and running time of each invocation of \Seau\ with parameters $0.1$ and $\delta'$ are
	$\tilde{O}\left(\frac{n}{\sqrt {m}} \cdot \poly\log(1/\delta')\right)$.
	Therefore, by the setting of $r$ and $\delta'$, the expected query complexity and running time are
	$\tilde{O}\left(\frac{n}{\sqrt {m}} \cdot \frac{\log(1/\eps)}{\eps} \right)\cdot \poly(\log(1/\delta))$.
\end{proofof}
\fi

\section{Adaptation of the algorithm when given a  $(1\pm \eps/c)$ estimate of $m$} \label{sec:given m}

In this section we describe an adaption of the algorithm for the case where it is given a  $(1\pm \eps/c)$ estimate of $m$.  It will be easier to think of every edge $\{u,v\}\in E$ as two distinct \emph{directed} edges $(u,v)$ and $(v,u)$.
We take advantage of the following definitions and simple claim.

\begin{definition}
	We say that a vertex $v$ is \emph{high} if $d(v)> 2\alpha/\eps$. Otherwise, we say it is \emph{low}.
	
	For an edge $(u,v)$, if  $u$ is low then we say it is a  \emph{directed low edge}, and otherwise we say that it is a \emph{directed high edge}.
\end{definition}

  \begin{claim}\label{clm:number_high}
 	If a graph $G$ has arboricity at most $\alpha$, then it has at most $2\eps m$ 
 	directed high edges.
 \end{claim}
 \begin{proof}
 	Let $H$ denote the set of high degree vertices in the graph. Then $|H|< 2m/(2\alpha/\eps)=\eps m/\alpha$, and it follows that $|E(H)|< \alpha |H|=\eps m$, implying that the number of  directed high edges is at most $2\eps m$.
 \end{proof}

We shall also make use of Lemma 3.1 of Eden and Rosenbaum~\cite{ER17-arxiv} for sampling directed light edges.\footnote{This theorem appears only in version \cite{ER17-arxiv}, and not in the final version of the same paper~\cite{ER17}.}

 For a degree threshold $\theta$ let $E_{\leq \theta}$ denote the set of directed edges $(u,v) \in E$ such that $d(u)\leq \theta$.

\begin{lemma}[Eden \& Rosenbaum \cite{ER17-arxiv}] \label{lem:Sle}
	Let $G = (V, E)$ be a graph with $n$ vertices and $m$ edges.
	There exists a procedure named \Sle\ that given $\theta$, returns a directed edge in $E_{\leq \theta}$ with probability $\frac{|E_{\leq \theta}|}{n \cdot \theta}$. Furthermore, the procedure performs a constant number of queries and each directed edge in $E_{\leq \theta}$ is returned with equal probability.
\end{lemma}

\renewcommand{\setr}{\frac{n\alpha}{\om} \cdot \frac{200\ln(\tnew{1}/\delta)}{\eps^{\danachange{3}}}}
\newcommand{\Ehe}{\hyperref[Ehe]{\color{black} \bf Estimate-High-Edges}}
\newcommand{\ml}{m_{low}}
\newcommand{\mh}{m_{high}}
\newcommand{\oml}{\overline{m}_{low}}
\newcommand{\omh}{\overline{m}_{high}}

\ifnum\TALG=0
\confFigure{
	{\bf \Ehe$(n,\alpha,\eps,\delta, \bar{m})$} \label{Ehe}
	\smallskip
	\begin{compactenum}
		\item Set $r=\setr$.
		\item For $i=1$ to $r$ do: \label{step:ehe_loop}
		\begin{compactenum}
			\item Invoke the procedure \Sle\
\danachange{with $\theta = 2\alpha/\eps$} and if a directed edge $(u,v)$ was returned such that $u \prec v$, then let $\chi_i=1$. Otherwise, let $\chi_i=0$.
		\end{compactenum}
		\item Let $\chi=\frac{1}{r}\sum_{i=1}^r \chi_i$.

		\item Let $\oml = \frac{2n \cdot \alpha}{\eps}\cdot \chi$ and let $\omh= \om-\oml$.
		\item If  $\omh > 5\eps \om /2$, then return \MANY. Otherwise, return \FEW.
	\end{compactenum}
}
\else
\floatname{algorithm}{Procedure}
\begin{algorithm}
\caption{\Ehe$(n,\alpha,\eps,\delta, \bar{m})$} \label{Ehe}
\begin{algorithmic}[1]
\STATE Set $r=\setr$ .
\STATE Invoke the procedure \Sle\ for $r$ times with $\theta = 2\alpha/\eps$ and let $\chi_i=1$ if the $i\th$ invocation returned an edge, and otherwise let $\chi_i=0$.
\STATE Let $\chi=\frac{1}{r}\sum_{i=1}^r \chi_i$.
\STATE Let $\oml = \frac{2n \cdot \alpha}{\eps}\cdot \chi$ and let $\omh= \om-\oml$.
\STATE If  $\omh > 5\eps \om /2$, then return \MANY. Otherwise, return \FEW.

\end{algorithmic}
\end{algorithm}
\floatname{algorithm}{Algorithm}
\if

\begin{claim}\label{clm:Ehe}
Assume that  $\om \in (1\pm \eps/4)m$. If there are more than $3\eps m$ directed high edges in the graph, then with probability at least $1-\delta$, \Ehe\ returns \MANY, and if there are at most $2\eps m$ directed high edges in the graph, then with probability at least $1-\delta$, \Ehe\ returns \FEW. The query complexity and running time of the procedure are
$O\left(\frac{n\alpha\cdot \log(1/\delta)}{\eps^{\danachange{3}}\cdot m}\right)$.
\end{claim}
\begin{proof}
Let $\ml$ and $\mh$ denote the number of directed low and high edges in the graph, respectively.
We first consider the case that $G$ has more than $3\eps m$ high edges, so that $\ml \leq (1-3\eps)m$.
By the above and Lemma~\ref{lem:Sle}, $\EX[\chi]=\frac{\ml}{(2n\alpha/\eps)}\leq \frac{(1-3\eps)m}{(2n\alpha/\eps)} \;.
$
Therefore, by the multiplicative Chernoff bound, the setting of $r=\setr$ in the algorithm, the assumption
\danachange{that $\om \leq (1+\eps/4)m$, and the assumption that $\eps \leq 1/20$},
\begin{eqnarray}
\Pr\left[\chi > \left(1+\frac{\eps}{4}\right)\cdot \frac{(1-3\eps)m}{(2n\alpha/\eps)}\right] &<&
\exp\left(-\frac{\eps^2 \cdot \frac{(1-3\eps)m}{(2n\alpha/\eps)}\cdot r }{16\cdot 3}\right) \nonumber \\
&=& \danachange{ \exp\left(\frac{\eps^3(1-3\eps) m}{96 n\alpha}\cdot \setr \right)}
\;<\; \delta 
\;.
\end{eqnarray}

Hence, with probability at least $1-\delta$,
	$$\chi \leq \frac{(1+\eps/4)\cdot (1-3\eps)\eps m}{2n \cdot \alpha} <
	\frac{(1-5\eps/2) \cdot \eps\om}{2n \cdot \alpha}$$ and $\oml < (1-5\eps/2) \cdot \om\;.$
It follows that  $\omh > 5\eps \om/2 $ with probability at least $1-\delta$, so that the algorithm will return \MANY.

We now consider the case that $G$ has less than $2\eps m$ high edges, so that $\ml > (1-2\eps)m $. Hence, by Lemma~\ref{lem:Sle}, $\EX[\chi]> \frac{(1-2\eps)m}{(2n\alpha/\eps)}$.
\tnew{By the multiplicative Chernoff bound, the setting of $r$, the assumption
that $\om \leq (1+\eps/4)m$ and the assumption that $\eps \leq 1/20$},
\[
\Pr\left[\chi < \left(1-\frac{\eps}{4}\right)\cdot \frac{(1-2\eps)m}{(2n\alpha/\eps)}\right] <
\exp\left(-\frac{\eps^2 \cdot \frac{(1-2\eps)m}{(2n\alpha/\eps)}\cdot r }{16\cdot 2}\right) < \delta \numberthis \label{eqn:chi}
\;.\]

\danachange{It follows that} with probability at least $1-\delta$,
$$\chi\geq\frac{(1-\eps/4)(1-2\eps)\eps m}{2n \cdot \alpha} > \frac{(1-5\eps/2)\cdot \eps \om }{2n \cdot \alpha}. $$
Therefore, with probability at least $1-\delta$, $\danachange{\omh} < 5\eps \om/2$ and the procedure returns \FEW.

By Lemma~\ref{lem:Sle}, the query complexity and running time of each invocation of the procedure \Sle\
\danachange{(with $\theta = 2\alpha/\eps$)}  are $O(1)$. Hence, the query complexity and running time of the procedure are  $O(r)=O\left(\frac{n\alpha\cdot \log(1/\delta)}{\eps^{\danachange{3}}\cdot m}\right)$.
\end{proof}

	 Consider modifying the algorithm \Iba\ (from Section~\ref{sec:test-arb}) as follows. We first check if there are many high edges in the graph, and if so we reply that the graph is far from having arboricity at most $3\alpha$. Otherwise, we sample light edges and check if their endpoints are active.

 \newcommand{\tapp}{\frac{1000n \alpha}{\eps^{\tnew{2}} \om}}
 \newcommand{\sapp}{\frac{800}{\eps}}
\ifnum\TALG=0
 \confFigure{
 	{\bf \Ibgm$(G,\alpha,\eps,\om)$} \label{Iba}
 	\smallskip
 	\begin{compactenum}
 			 \item Invoke \Ehd$(n,\alpha,\eps, 1/12,\om)$, and if the procedure returns \MANY\ then return \NO. \label{step:Ehe}
			\item Invoke 
\Sle\ \danachange{with $\theta = 2\alpha/\eps$} for $t=\tapp$ times, and let $S$ be the (multi-)set of returned \emph{directed} edges. Let $s$ be the number of (not necessarily different) edges in $S$. \label{step:sample}
			\item If $s<\sapp$ then return \NO. \label{step:small_s}
			\item Set $\ell = \lceil\log_{6/5}(1/\eps)\rceil$
			.
			\item For every directed edge $(u_i, v_i) \in S$ do:
\begin{compactenum}
	\item Invoke \IsActive$(u_i,\ell,\danachange{\alpha,\eps},\eps/2)$ and \IsActive$(v_i,\ell,\danachange{\alpha,\eps},\eps/2)$.
	If the procedure returned \YES\ on both invocations, then set $\chi_i=1$. Otherwise, set $\chi_i=0$.
\end{compactenum}
\item Set $\chi=\frac{1}{s}\sum_{i=1}^s \chi_i$.
\item If $\chi <10\eps m$, then return \YES. Otherwise, return \NO.
 	\end{compactenum}
 }
\else
\begin{algorithm}
\caption{\bf \Ibgm$(G, n,\alpha,\eps,\om)$} \label{Ibagm}
\begin{algorithmic}[1]
	\STATE Invoke \Ehd$(n,\alpha,\eps, 1/12,\om)$, and if the procedure returns \MANY\ then return \NO. \label{step:Ehe}
	\STATE Invoke 
 \Sle\ \danachange{with $\theta = 2\alpha/\eps$} for $t=\tapp$ times, and let $S$ be the (multi-)set of returned \emph{directed} edges. Let $s$ be the number of (not necessarily different) edges in $S$. \label{step:sample}
	\STATE If $s<\sapp$ then return \NO. \label{step:small_s}
	\STATE Set $\ell = \lceil\log_{6/5}(1/\eps)\rceil$.
	\FOR{every directed edge $(u_i, v_i) \in S$}
		\STATE Invoke \IsActive$(u_i,\ell,\danachange{\alpha,\eps},\eps/2)$ and \IsActive$(v_i,\ell,\danachange{\alpha,\eps},\eps/2)$.
		If the procedure returned \YES\ on both invocations, then set $\chi_i=1$. Otherwise, set $\chi_i=0$.
	\ENDFOR
	\STATE Set $\chi=\frac{1}{s}\sum_{i=1}^s \chi_i$.
	\STATE If $\chi <12\eps$, then return \YES. Otherwise, return \NO.
\end{algorithmic}
\end{algorithm}
\fi

 \begin{theorem}\label{thm:test-given-m}
 	Assume that $\om\in (1\pm\eps/4)m$. If $G$ is $\eps$-close to having arboricity at most  $\alpha$, then  \Ibgm\ returns {\rm \YES} with probability at least $2/3$, and if $G$ is
 	$20\eps$-far from having arboricity at most $3\alpha$, then \Ibgm\ returns \NO\ with probability at least $2/3$.
 	
 	The query complexity and running time of the algorithm are
 $$\tilde{O}\left(\frac{n\alpha}{\eps^{\danachange{3}} m}  +
 	\left(\frac{1}{\eps}\right)^{O(\log(1/\eps))}\right)$$
 	in expectation.  
 \end{theorem}

 \begin{proof}
 	
 	We first consider the case that $G$ is $\eps$-close to having arboricity at most $\alpha$, and prove that the algorithm returns \YES\ with probability at least $2/3$.
 	If $G$ is $\eps$-close to having arboricity at most $\alpha$, then it follows from Claim~\ref{clm:number_high} that $G$ has less than $2\eps m$ directed high edges. Therefore, by Claim~\ref{clm:Ehe}, the procedure \Ehe\ returns \FEW\ with probability at least $11/12$. We henceforth condition on this event.

By Lemma~\ref{lem:Sle}\  each invocation of
\Sle\ with $\theta = 2\alpha/\eps$
returns a directed light edge with probability at least $ \frac{(1-2\eps)m}{(2n\alpha/\eps)}$.
\tnew{Let us denote this probability by $p_{succ}$ and let $x_i$ be a random variable that indicates if the $i\th$ invocation of \Sle\ returns an edge. By the multiplicative Chernoff bound,
\[\Pr\left[\frac{1}{t}\sum_{i=1}^t x_i < 0.9 p_{succ}\right] < \exp\left(-\frac{0.1^2\cdot p_{succ}\cdot t}{2}\right) \leq \exp\left(-\frac{1}{200}\cdot \frac{ (1-2\eps)\cdot \eps m}{2n\alpha} \cdot \tapp\right) < \frac{1}{12} \;,\]
where the last inequality is by the assumption that $\om \leq  (1 + \eps/4)m$ and that $\eps\leq 1/20$.
Therefore, with probability at least $11/12$, $s>0.9\cdot p_{succ}\cdot t > \sapp$.}
 Condition on this event as well.
 	
 	By Lemma~\ref{lem:det-assign-ub}, if $G$ is $\eps$-close to having arboricity at most $\alpha$, then $m(\rG) \leq 5\eps m$, so that $\ml(\rG) \leq 10\eps m$, where
 \danachange{for a subgraph $G'$, we let $\ml(G')$ denote the number of directe low edges in $G'$}.
 		For every $i$ such that $(u_i, v_i)$ is not in $\rG$,  it holds that either $u_i$ or $v_i$ is not in $\Ualpha$. Hence, by Lemma~\ref{lem:is-active}, \IsActive\ returns \YES\ on both vertices with probability at most $\eps/2$.
	For every $i$ such that $(u_i, v_i)$ is in $\rG$, we bound the probability that \IsActive\ returns \YES\ on both vertices by $1$. Since by Lemma~\ref{lem:Sle}, each directed light edge in the graph is returned with equal probability, it holds that
	\[\EX[\chi_i] \leq \frac{(\eps/2)\cdot (\ml(G)-\ml(\rG))}{\danachange{\ml}} + \frac{\ml(\rG)}{\danachange{\ml(G)}}\leq \frac{\eps}{2}+\frac{10\eps m}{\ml(G)}\leq 11.7\eps \;,\]
	where the last inequality is by the fact that $\ml\geq (1-2\eps)m$ \tnew{and the assumption that $\eps<1/20$}.
Therefore, by the multiplicative Chernoff bound, and since  $s>\sapp$,
\begin{align*}  \Pr\left[ \frac{1}{s}\sum_{i=1}^s \chi_i > \left(1+\frac{1}{40}\right)\cdot 11.7\eps \right]
<	\exp\left(-\frac{(1/40)^2\cdot 11.7\eps \cdot s}{3}\right) < 1/6.
\end{align*}
Therefore with probability at least $5/6$, $\chi <12\eps$.
By taking a union bound over all ``bad'' events, it holds that the procedure returns \YES\ with probability at least $2/3$.

Now we consider the case that $G$ is at least $20\eps$-far from having arboricity at most $3\alpha$, and prove that with probability at least $2/3$, the algorithm returns \NO. If $G$ has more than $3\eps m$ high edges, then by Claim~\ref{clm:Ehe}, with probability at least $11/12$ the procedure \Ehe\ will return \MANY\ in step~\ref{step:Ehe}, and therefore the algorithm will return \NO\ and we are done.
Also, if $s<\sapp$, then the procedure returns \NO\ in Step~\ref{step:small_s}, and we are done. Therefore, assume that  $G$ has at most $3\eps m$ high edges and that $s\geq \sapp$.

By Lemma~\ref{lem:det-assign-lb}, since $G$ is at least $20\eps$-far from having arboricity at most $3\alpha$, it holds that $m(\Gmax)>16\eps m$, implying that $\ml(\Gmax)>13\eps m$.
For every $i$ such that $(u_i, v_i)$ is in $\Gmax$,  both
$u_i$ and $v_i$ are in $\Umax$, and by Lemma~\ref{lem:is-active} and the union bound, \IsActive\ returns \YES\ on both vertices with probability at least $1-\eps$.
Also, it follows from  Lemma~\ref{lem:Sle} that every directe light edge in the graph is returned with equal probability. Hence,
\[\EX[\chi_i]\geq \frac{(1-\eps)\cdot \ml(\Gmax)}{m}\geq 12.35\eps .\]
By the multiplicative Chernoff bound and since $s > \sapp $,
\begin{align*}
	\Pr\left[ \frac{1}{s}\sum_{i=1}^s \chi_i < \left(1-\frac{1}{40}\right)\cdot {12.35\eps} \right]
	< \exp\left(-\frac{(1/40)^2\cdot 12.35\eps \cdot s}{2}\right) < 1/6.
\end{align*}
Therefore, if $G$ is $20\eps$-far from having arboricity at most $3\alpha$, then with probability at least $2/3$, $\chi>  12\eps$ and the algorithm returns \tnew{\NO}.

By Claim~\ref{clm:Ehe}, the query complexity and running time resulting from the invocation of the procedure \Ehe\ in Step~\ref{step:Ehe} are $O(\frac{n\alpha}{\eps^{\danachange{3}} m})$.
By Lemma~\ref{lem:Sle}, the running time and query complexity of the procedure \Sle\ 
are constant, and therefore the query complexity and running time of
Step~\ref{step:sample} are $O\left(\frac{n\alpha}{\eps^{\tnew{2}} m} \right)$.
In each step of the for loop there are two invocations of the procedure \IsActive\ with parameters $\gamma=\eps$ and $\delta=\eps/2$. By Lemma~\ref{lem:is-active}, the query complexity and running time resulting from these invocations are
$O\left(\left(\frac{1}{\eps}\right)^{O(\log(1/\eps))}\right)$.
Therefore, the total query complexity and running time are
$O\left(\frac{n\alpha }{\eps^{\danachange{3}} m}+
\left(\frac{1}{\eps}\right)^{O(\log(1/\eps))}\right)$.
 \end{proof}

\end{document}